\documentclass[twoside,11pt]{article}
\usepackage{amsmath}
\usepackage{bbm}
\usepackage{graphicx}
\usepackage{algorithm} 
\usepackage{algorithmic}
\usepackage{subcaption}
\usepackage{multirow}

\usepackage[utf8]{inputenc} 
\usepackage[T1]{fontenc}    
\usepackage{booktabs}       
\usepackage{nicefrac}       
\usepackage{microtype}      
\usepackage{xcolor}         

\usepackage{jmlr2e}

\usepackage{lastpage}

\firstpageno{1}

\begin{document}

\title{Learning to Assign Prediction Tasks to Agents with Capacity Constraints}

\author{\name Shang Wu \email shangw13@uci.edu \\
       \addr Department of Computer Science\\
       University of California, Irvine\\
       Irvine, CA 92697, USA
       \AND
       \name Saatvik Kher \email sukher@uci.edu \\
       \addr Department of Computer Science\\
       University of California, Irvine\\
       Irvine, CA 92697, USA
       \AND
       \name Padhraic Smyth \email smyth@ics.uci.edu \\
       \addr Department of Computer Science\\
       University of California, Irvine\\
       Irvine, CA 92697, USA}

\maketitle

\begin{abstract}
We address the problem of learning to assign prediction tasks to one agent from a set of available agents, including human decision-makers and AI models. We focus on sequential learning of agent expertise and assignment policies where each agent is constrained to handle a fraction of tasks. We provide a general theoretical characterization of this problem in terms of agent capacities, differences in agent expertise, and task context. We then develop a framework of sequential explore-exploit policy-learning algorithms that seek to maximize overall performance. Experimental results over a variety of tabular, image, and text prediction tasks demonstrate systematic gains from our policy-learning algorithms relative to non-contextual baselines across different types of agents, including LLMs and humans.  
\end{abstract}

\begin{keywords}
  Human-AI collaboration; Assignment policy; Online learning
\end{keywords}

\section{Introduction}
In many real-world settings, tasks arrive sequentially and need to be allocated to agents, where the agents  (AI or human)  may differ in their expertise. An example is in the context of medical radiology, where the set of agents could include an AI model and multiple human radiologists with different experiences, each available for tasks such as providing a diagnostic class label for a radiology image. 
For many types of tasks, human expertise is often quite specific \citep{ericsson1993role, chase1973perception}. For example, clinicians tend to perform well on familiar types of medical cases but do not generalize equally across all cases \citep{eva2005every,norman2007non}. AI models also can differ in terms of their expertise, for example, in large language model (LLM) routing where the focus is on determining which LLM is best-suited to answering different queries  \citep{srivatsa2024harnessing,song2025irt}. Further, in the context both human and AI expertise, the value of human-AI complementarity is widely acknowledged \citep{gonzalez2026toward}, across applications such as medical diagnosis \citep{zoller2025human}, image classification \citep{steyvers2022bayesian},  and content moderation \citep{lai2022human}. These findings suggest that different types of prediction agents can have localized expertise, and that effective task assignment should account for heterogeneous, context-dependent performance.

Another component of many real-world task assignment problems is the presence of {\bf capacity constraints} on how much workload an individual agent can handle, i.e., what fraction of tasks in the long-run they are assigned. These are distinct from cost constraints: capacity constraints impose agent-specific workload limits, rather than only limiting aggregate spending. For example, in the radiology example earlier, assigning many or all tasks to the most senior/experienced human radiologist is typically not a practical option (e.g., due to human fatigue or burnout \citep{berlin2000liability, alzoubi2024moderating}). Similarly, assigning all tasks to an AI agent (and none to humans) may be sub-optimal in terms of leveraging relevant human expertise and may be undesirable from a safety and employee satisfaction \citep{salikutluk2024evaluation,zhang2022you,shneiderman2020human}.

Our work differs from prior work in two primary aspects. 
First, we jointly learn assignment policies and agent expertise online (useful, for example, under distribution shift), unlike prior work on learning to defer or LLM routing, which learns agent expertise offline \citep{mozannar2020consistent, claure2020multi, chen2023frugalgpt}. 
Second, we study this problem under capacity constraints on the agents, unlike fairness-constrained bandits, which enforce a shared minimum-rate constraint across arms \citep{chen2020fair}, and cost-constrained LLM routing, which enforces a global budget rather than agent-specific capacity limits \citep{panda2025adaptive,li2026wiserouter,taberner2026paretobandit}.
Our technical approach builds on contextual multi-armed bandit frameworks, with tasks arriving sequentially and the assignment policy and agent expertise models being updated from observed rewards.  Across simulated agents, real human-annotation data, machine learning classifiers, and LLMs, we show how contextual assignment outperforms non-contextual assignment.\footnote{In this paper, we use ``non-contextual assignment'' and ``random assignment'' interchangeably.} More broadly, our framework highlights the value of jointly modeling latent expertise and capacity constraints in sequential decision-making environments, such as human-AI collaboration and large-scale organizational systems.
Our work makes the following contributions:
\begin{itemize}
    \item We provide a theoretical framework for capacity-constrained task assignment, introducing for example the concept of a {\it contextual reward gap} between agents as a general framework for providing insights into task allocation and for characterizing optimal policies.
    \item We conduct a series of experiments across settings with human annotators, machine learning classifiers, and  LLMs, including distribution shift, showing systematic gains for contextual capacity-constrained policies over non-contextual baselines.\footnote{We will release the code, experiment scripts, and datasets in the final version of the paper.}
\end{itemize}

\section{Optimal Assignment Policies and Capacity Constraints}
Below we define the assignment problem and characterize the structure of optimal policies in an oracle setting where agents' context-dependent expertise is assumed known. This allows us to isolate how contextual assignment and capacity constraints shape the optimal allocation rule. 
Later in Section~\ref{sec:learning_policy}, we address the practical problem of online learning of expertise.
\vspace{-8pt}

\subsection{Tasks, Agents, and Rewards}
\label{subsec:task_agent_reward} 
\vspace{-5pt}
There are $A$ prediction agents indexed by $a \in \{1,\dots,A\}$, where agents can be machine learning models or humans. We assume that each prediction task needs to be assigned to a single agent to make a prediction.  
Prediction tasks arrive sequentially, where each task is characterized by an observable {\bf context} $x_t  \in \mathcal{X} \subseteq \mathbb{R}^d$,  $t=1,\ldots,T$, and contexts are drawn i.i.d.\ from some unknown distribution $P_X$. 
If assigned a task, a prediction agent $a$ (a model or a human) generates a prediction $\hat{y}_{a,t}$. In what follows, we will focus on classification tasks and agents, where $\hat{y}_{a,t}$ is from a set of $K$ possible labels, but the framework we discuss can be directly extended to other prediction problems such as regression. The agent's prediction $\hat{y}_{a,t}$ can be based not only on the context $x_t$ but on features beyond the context (e.g., an agent could use all pixels in an image as input, whereas the context could be some lower-dimensional representation of the pixel information).     
 
Assigning task $t$ to agent $a$ yields a stochastic reward $r_{t,a} \sim P_a(\cdot \mid x_t)$, where $P_a(\cdot \mid x)$ denotes the agent-specific reward distribution conditioned on context $x_t$. Rewards reflect the prediction quality of the agent. A natural choice for reward in a classification context (and one we will focus on in this paper) is the accuracy of agent $a$'s prediction for task $t$, i.e., $r_{t,a} = \mathbbm{1}\{\hat{y}_{a,t} = y_t\}$.
Here, $y_t$ is the true observed label for the task, observed after a prediction is made\footnote{While we assume immediate feedback in our main setting, the framework naturally extends to settings with delayed rewards, as standard contextual bandit methods can accommodate delayed updates \citep{blanchet2024delay}.}. Alternative reward definitions (not explored here) could be negative log-loss or negative squared error (Brier score).

We define the {\bf conditional expected reward} of agent $a$ as
$ \mu_a(x) := \mathbb{E}[r_{t,a} \mid x_t = x] $, which captures agent $a$’s performance conditioned on context $x$, with the corresponding {\bf marginal reward} defined as $\mu_a := \mathbb{E}_{x \sim P_X}[\mu_a(x)]$, representing agent $a$'s overall performance with respect to the task/context distribution. Using accuracy (as above) as the definition of reward, these expectations correspond to conditional expected accuracy as a function of context $x$ and marginal accuracy. 
To illustrate these ideas, Figure~\ref{fig:llm_accuracy} shows a simple case of two LLM classifier agents. The x-axis represents context $x$ and the $y$-axis corresponds to the estimated empirical conditional accuracy  for each agent, with varying agent expertise (accuracy, or expected reward) as a function of context. 

\begin{figure}[t]
    \centering
    \includegraphics[width=0.65\linewidth]{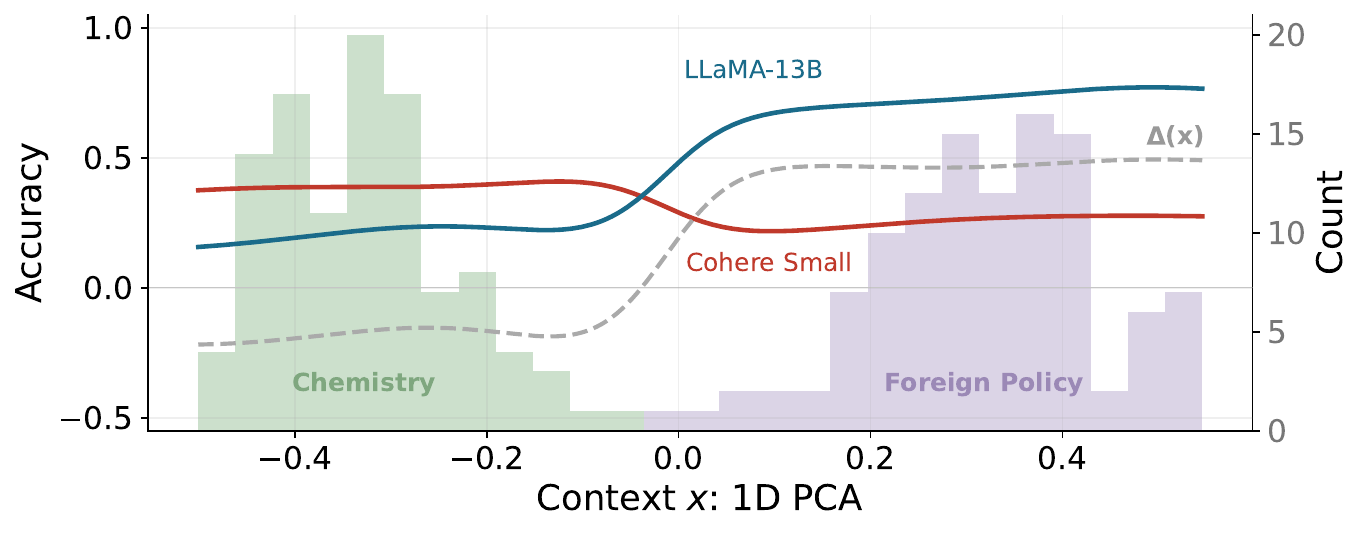}
    \caption{
    Accuracy (smoothed, empirically estimated) of two LLMs as a function of context. Context is defined as the first principal component of embeddings of multiple-choice questions from MMLU \citep{hendrycks2020measuring}, with embeddings from \texttt{sentence-transformers/all-MiniLM-L6-v2}. Cohere Small performs better on College Chemistry tasks (left), while LLaMa-13B performs better on US Foreign Policy tasks (right). The dashed gray curve denotes the estimated contextual reward gap $\Delta(x):=\mu_{\mathrm{LLaMa}}(x)-\mu_{\mathrm{Cohere}}(x)$. The marginal accuracies are 0.32 for Cohere Small and 0.47 for LLaMa-13B.
    }
    \label{fig:llm_accuracy}
\end{figure}

\subsection{Optimal Assignment Policies without Agent Constraints}

\begin{definition}[Assignment Policy]
\label{def:assignment_policy}
    We define an {\bf assignment policy} $\phi$ as a (possibly stochastic) mapping from a task $x$ to an agent $a \in \{1,\ldots,A\}$. Let $P_{\phi}(a \mid x)$ denote the probability that policy $\phi$ assigns a task to agent $a$ given context $x$. For deterministic policies we have $P_{\phi}(a \mid x) = 1$ for a single agent $a = \phi(x)$ and $0$ otherwise. The {\bf marginal reward} of a policy $\phi$ is
    \[\mu^{\phi} := \mathbb{E}_{x \sim P_X} \left[ \sum_{a=1}^A P_{\phi}(a \mid x)\mu_a(x) \right].\]
\end{definition}

As a baseline, consider a non-contextual (random) policy that ignores $x$ and assigns tasks using some fixed set of probabilities $\{P_a\}$, where $\sum_a P_a = 1$, resulting in $\mu^{\phi} = \sum_{a=1}^A P_a \mu_a$.
Similarly, a policy that always assigns tasks to a single fixed agent $a'$ (such as the agent with the highest marginal reward) corresponds to $P_{a'} = 1$, resulting in marginal reward $\mu^{\phi} = \mu_{a'}$. 
For context-dependent deterministic policies, $\phi$ can be represented via scoring functions $\{s_a(x)\}_{a=1}^A$ such that
$\phi(x) = \arg\max_{a} \{ s_a(x) \}$.
This partitions the context space into regions $\mathcal{R}_a = \{x : \phi(x)=a\}$. The marginal reward averages over the context distribution $P_X$, so $\mu^{\phi} = \sum_{a=1}^A \int_{\mathcal{R}_a} \mu_a(x)\, P_X(x)\, dx$. 

\begin{definition}
    The {\bf optimal (unconstrained) policy} is $ \phi^\star(x) = \arg\max_{a} \mu_a(x)$, i.e., it assigns each context $x$ to the agent with the highest conditional expected reward for that context. Equivalently, the optimal decision regions $\{\mathcal{R}_a^\star\}$ satisfy $\mathcal{R}_a^\star = \{x \in \mathcal{X} : \phi^\star(x) = a\}$.
\end{definition}
\vspace{-15pt}

\paragraph{Two-agent case without constraints:}
As mentioned earlier, Figure~\ref{fig:llm_accuracy} illustrates heterogeneous expertise: different agents perform better in different regions of the context space, indicating that the optimal policy requires assigning tasks based on context. To formalize this, consider two agents with context-dependent performance $\mu_1(x)$ and $\mu_2(x)$, and define the {\bf contextual reward gap} $\Delta(x) := \mu_1(x) - \mu_2(x)$; see the dashed gray curve in Figure~\ref{fig:llm_accuracy} for an empirical estimate. The optimal (unconstrained) policy assigns each task to the agent with higher expected reward, $\phi^\star(x) = \arg\max_{a \in \{1,2\}} \mu_a(x)$, 
equivalent to a threshold rule in $\Delta(x)$.
In contrast, a context-independent policy selects a single agent based on marginal performance, yielding reward $\max\{\mu_1, \mu_2\}$, where $\mu_a = \mathbb{E}_x[\mu_a(x)]$. This policy ignores variation in $\Delta(x)$ and is generally suboptimal. 

Gains by a context-dependent policy arise whenever $P_X(\Delta(x)>0)>0$ and $P_X(\Delta(x)<0)>0$ (see Appendix \ref{app:two_agent}). A special case arises when one agent dominates across all contexts, i.e., $\mu_{a^\star}(x) \ge \mu_a(x)$ for all $a$ and $x$, in which case assigning all tasks to $a^\star$ is optimal. For context-dependent policies, the performance gap over a non-contextual policy is driven by regions where $\Delta(x)$ has the same sign, i.e., regions in the context space where one agent strictly dominates the other. 

\subsection{Optimal Assignment Policies with Agent Capacity Constraints}
\label{subsec:cap_constraint}
We now extend the framework to incorporate long-run capacity constraints across $A$ agents, meaning that each agent's capacity target is enforced in expectation over time.
Let $\alpha_a \in [0,1]$ denote the target long-run fraction of tasks assigned to agent $a$, with $\sum_{a=1}^A \alpha_a = 1$. 
The objective is to maximize expected reward subject to these long-run capacity constraints. We assume that the capacity constraints $\alpha_a$ are known and fixed for the problem, e.g., a company or hospital might have a policy that 90\% of tasks are assigned to a machine learning model and 10\% to a human expert. This notion of capacity differs from the cost-budget constraints studied in \citet{panda2025adaptive}: a cheaper, better-performing agent cannot simply be used for all tasks, because each agent has its own long-run assignment target.
For a non-contextual assignment policy, the only option is to assign tasks (randomly) based on probabilities.
$P_a=\alpha_a$, yielding $\mu^{\mathrm{rand}}=\sum_{a=1}^A \alpha_a \mu_a$.
For the theoretical characterization below, we assume a continuous context distribution: deterministic threshold policies achieve the target assignment capacities exactly. With discrete contexts, the same characterization applies, where exact equality requires stochastic tie-breaking at boundary contexts.

\begin{definition} 
    The {\bf optimal  constrained policy} 
    maximizes $\mu^{\phi} = \mathbb{E}_{x}\!\left[\sum_{a=1}^A P_\phi(a \mid x)\mu_a(x)\right]$ subject to capacity constraint $\mathbb{E}_{x}[P_\phi(a \mid x)] = \alpha_a$, 
    $ a = 1,\ldots, A$, $A \ge 2$.  
\end{definition}

\begin{theorem}[General Form of Optimal Constrained Policy]\label{thm:optimal_constrained}
The optimal policy takes a simple form: there exist constants $\{\lambda_a\}_{a=1}^A$ such that
\[
\phi^\star(x) = \arg\max_{a \in \{1,\dots,A\}} \bigl\{ \mu_a(x) - \lambda_a \bigr\}.
\]
\end{theorem}
See Appendix~\ref{app:general_case} for proof and additional details.
In the optimal policy, each task is assigned to the agent with the highest adjusted reward (i.e., $\mu_a(x) - \lambda_a$), where $\lambda_a$ acts as a shadow price for capacity. This policy is at least as good as the non-contextual assignment (i.e., $\mu^{\phi^\star} \ge \mu^{\mathrm{rand}}$), 
with strict inequality whenever the identity of the best agent depends on the context $x$.  In practice, both the reward functions $\mu_a(x)$ and the shadow prices $\{\lambda_a\}$ are unknown. In Section~\ref{sec:learning_policy}, we introduce a queue-based method that jointly learns $\mu_a(x)$ and approximates $\lambda_a$, $a = 1,\ldots, A$.

\paragraph{Two-agent case with capacity constraints.}
We consider a simplified setting with two agents, with capacities $\alpha_1=\alpha$ and $\alpha_2=1-\alpha$, with contextual reward gap $\Delta(x)  = \mu_1(x) - \mu_2(x)$. 

\begin{proposition}[Two-Agent Optimal Constrained Policy]
    Under capacity constraints, the optimal policy assigns agent $1$ to the $\alpha$-mass of contexts with the largest values of $\Delta(x)$, and agent $2$ to the remainder, yielding a threshold rule in $\Delta(x)$. See Appendix~\ref{app:two_agent_optimal_policy} for details. 
\end{proposition}
\label{prop:non-dominant}

In the non-dominant case, where both $P_X(\Delta(x)>0)>0$ and $P_X(\Delta(x)<0)>0$, agents outperform each other in different regions of the context space. For example, Figure~\ref{fig:llm_accuracy} shows that the two LLMs specialize in different regions of context space, so $\Delta(x)$ changes sign. In this example, if LLaMa-13B has 25\% capacity and Cohere Small has 75\%, then the optimal capacity-constrained policy would intuitively threshold the context around $x=0.3$. 
We also note that the combined agents’ performance can outperform either individual agent under certain conditions. In particular, when $\alpha = P_X(\Delta(x)\ge 0)$, the policy assigns each context to the better-performing agent, and the resulting reward satisfies $\mu^{\phi^\star} > \max\{\mu_1,\mu_2\}$ (see Appendix~\ref{app:two_agent_outperform}).

More generally, the gains from contextual assignment are driven by disagreement between agents. Let $\mathcal{D}_1 = \{x:\Delta(x)>0\}$ and $\mathcal{D}_2 = \{x:\Delta(x)<0\}$ denote the regions where each agent has an advantage. The improvement over non-contextual assignment depends on both the size of these {\bf disagreement regions} and the magnitude of $\Delta(x)$, as well as the capacity  $\alpha$. 

\begin{proposition}[Full Exploitation of Disagreement]
\label{prop:full_disagreement}
    When capacity allows full exploitation of disagreement, i.e., $p_1 \le \alpha \le 1 - p_2$, the optimal policy assigns all $x \in \mathcal{D}_1$ to agent $1$ and all $x \in \mathcal{D}_2$ to agent $2$.
\end{proposition}

In this case, the gain over non-contextual assignment is $\mu^{\phi^\star} - \mu^{\mathrm{rand}}
=
(1-\alpha)p_1
\mathbb{E}\bigl[\Delta(X)\mid X\in\mathcal{D}_1\bigr]
+
\alpha p_2
\mathbb{E}\bigl[-\Delta(X)\mid X\in\mathcal{D}_2\bigr]$,
where $p_1 = P_X(\mathcal{D}_1)$ and $p_2 = P_X(\mathcal{D}_2)$. Thus, the magnitude of improvement depends not only on the size of the disagreement regions but also on the average contextual reward gap within those regions. This formalizes and generalizes the disparity intuition in \citet{chen2020fair}: gains arise when agents have different relative advantages across contexts, and our characterization links those gains to the size of the disagreement regions and the capacity $\alpha$. See Appendix~\ref{app:two_agent_disagreement} for details. 

\begin{proposition}[Optimal Policy under Agent Dominance]
    When one agent is dominant, i.e., $\mu_1(x)\ge \mu_2(x)$ for all $x$, the optimal policy assigns agent $1$ to the $\alpha$-mass of contexts with the largest values of $\Delta(x)$. As $\alpha \to 1$, the selected region $\mathcal{R}_1^\star$ expands to the entire context space $\mathcal{X}$, and the resulting reward satisfies $\mu^{\phi^\star}\to \mu_1$, converging to the unconstrained optimum.
    \end{proposition}

Since $\mu_1(x)\ge \mu_2(x)$ for all $x$, we have $\Delta(x)\ge 0$ everywhere. Because $\Delta(x)$ is nonnegative for all $x$, the threshold characterization corresponds to selecting the contexts where agent $1$ has the greatest advantage over agent $2$. As a specialization to the two-agent setting, the contextual policy satisfies $\mu^{\phi^\star} \ge \mu^{\mathrm{rand}}$ in both the non-dominant and dominant cases, with strict inequality whenever agent performance depends on the context (see Appendix~\ref{app:two_agent_optimal_over_radnom}).

\section{Online Learning of Assignment Policies with Agent Capacities}
\label{sec:learning_policy}
Below we develop a Bayesian multi-arm bandit (MAB) framework for learning the optimal assignment policy online when the conditional expected rewards $\mu_a(x)$ are unknown. At each round $t$, a context $x_t$ arrives, an agent $a_t$ is selected by the current policy (taking constraints into account), a reward $r_{t,a_t}$ is observed, and the contextual reward model is updated for the selected agent.

\subsection{Context Modeling}
For context models $\mu_a(x)$, we consider logistic and tree-based models as described below. For the {\bf logistic context model},
agent $a$'s conditional expected reward is modeled via $\mu_a(x_t) = \mathbb{P}(r_{t,a}=1 \mid x_t) = \sigma(\theta_a^\top x_t)$, where $\theta_a$ is a latent parameter vector governing agent $a$'s context-dependent expertise and $\sigma(z)=(1+e^{-z})^{-1}$. Following \citet{chapelle2011empirical}, we maintain an approximate Gaussian posterior for $\theta_a$ over rounds $t$ using a Laplace approximation. We investigate two MAB strategies in this context (alternatives such as the upper-confidence bound (UCB) method could also be used). The first is Thompson sampling (TS), where at each round $t$, a parameter vector is sampled from the current posterior distribution over $\theta_a$,  the corresponding sampled reward estimate $\tilde{\mu}_{a,t}(x_t)=\sigma(\tilde{\theta}_{a,t}^{\top}x_t)$ is computed, and this value is used by the assignment policy to select an agent.  The second strategy is greedy, using an MAP point estimate $\hat{\theta}_{a,t}$ to compute $\hat{\mu}_{a,t}(x_t)=\sigma(\hat{\theta}_{a,t}^{\top}x_t)$. The greedy policy selects agents based on this estimated expected reward, without explicit posterior sampling. See details in Appendix~\ref{app:logistics_posterior_update}

To capture non-linear dependencies in $\mu_a(x_t)$,
we also use {\bf tree-based (random forest) context models}, where $\mu_a(x_t)$ is estimated by averaging predictions across an ensemble of trees trained on a set of $B$ bootstrap samples. Again we investigate both greedy selection and Thompson sampling for assignment. 
For greedy selection, we use the ensemble mean $\mu_{a,t}(x_t) = \frac{1}{B}\sum_{b=1}^B f_{a,t}^{(b)}(x_t)$, where $f_{a,t}^{(b)}$ denotes the $b$-th tree for agent $a$ at time $t$. To enable exploration, we approximate Thompson sampling using bootstrap uncertainty (similar to \citep{eckles2014thompson}): at each round $t$, we sample a tree uniformly at random and use its prediction $\tilde{\mu}_{a,t}(x_t) = f_{a,t}^{(b)}(x_t), b \sim \mathrm{Uniform}\{1,\dots,B\}$. Additional details are in Appendix \ref{app:tree_contextual_model}.

\subsection{Enforcing Capacity Constraints}
\label{subsec:virtual_quque}
To enforce the constraints online, we adopt a queue-based approach following \citet{neely2010stochastic}; similar queue-based ideas have appeared in prior work on bandits with fairness constraints \citep{li2019combinatorial,huang2020thompson} and human-AI service systems \citep{montanari2026learning}. For each agent $a$, we maintain a virtual queue $Q_{t,a}$ that tracks deviations between realized assignments and the target capacity $\alpha_a$. The queues evolve as
\begin{equation}
Q_{t+1,a} = \bigl[ Q_{t,a} + I_{t,a} - \alpha_a \bigr]_+, \quad a=1,\dots,A,
\label{eq:virtual_queue}
\end{equation}
where $I_{t,a}:=\mathbbm{1}\{a_t=a\}$ indicates whether agent $a$ is selected at round $t$, and $[\cdot]_+ = \max(\cdot,0)$. Intuitively, $Q_{t,a}$ increases when agent $a$ is over-assigned relative to its target capacity and remains small otherwise.
At each round $t$, the assignment balances estimated reward and capacity pressure:
\begin{equation}
a_t = \arg\max_{a \in \{1,\dots,A\}} \left\{ s_a(x_t) \right\}, \qquad \mbox{where} \ s_a(x_t) = \mu_{a,t}(x_t) - \eta Q_{t,a},
\label{eq:agent_select_w_cap}
\end{equation}
where $\eta \ge 0$ controls the strength of the capacity penalty and $\mu_{a,t}(x_t)$ is generated either by  TS or the greedy procedure for each agent.
This rule can be interpreted as a dynamic approximation to the optimal policy of the form $\mu_a(x) - \lambda_a$, where $\{\lambda_a\}$ are the shadow prices arising from the constrained optimization problem (see Appendix~\ref{app:general_case}). $Q_{t,a}$ serves as a time-varying estimate of the shadow price. Intuitively, $\lambda_a$ acts as a penalty for assigning additional tasks to agent $a$, discouraging overuse when capacity is limited. When all agents have positive target capacities, this balancing mechanism also ensures that each agent continues to receive assignments and feedback, providing a form of implicit exploration even under greedy reward estimates.
Our general online learning framework, which combines contextual reward models with queue-based capacity constraints, is summarized in Algorithm~\ref{alg:online_assignment} in Appendix~\ref{app:context_model_para_update}. Implementation details are provided in Appendix~\ref{app:capacity_constraints}. Appendix~\ref{app:regret} gives a finite-time bound on regret against the constrained optimum of Theorem~\ref{thm:optimal_constrained}, equal to $\eta T$ plus the regret of the underlying contextual bandit learner, together with a deterministic $O(A/\eta)$ bound on the virtual queues.

\section{Experimental Results}
\label{sec:exp}
We evaluate our framework across settings involving different types of classification agents for $K$-ary classification tasks with binary rewards.
Our results demonstrate that contextual policies systematically outperform non-contextual baselines and individual agents across agent types and data modalities and also illustrate that capacity constraints induce enough implicit exploration for greedy policies to match Thompson Sampling.
We consider machine learning classifiers, LLMs, and both human and simulated human agents. Simulated agents allow for controlled evaluation of heterogeneous expertise and learning dynamics \citep{madras2018predict, mozannar2020consistent, alves2024cost}. 
Full details on experimental methods and datasets are provided in the Appendix. We focus here on the two-agent case ($A=2$)  for ease of visualization; multi-agent results are provided in Appendix~\ref{app:multi-agent}.

Across experiments, we induce distribution shift by training and evaluating on different data subsets, reflecting realistic deployment settings. Unless otherwise stated, all experiments follow the same evaluation protocol. We run each policy over 100 randomized permutations of the online task sequence and report average error rates. We vary the capacity of Agent 1 from 0 to 1, with Agent 2 taking the remaining capacity, and set the queue penalty parameter to $\eta=0.5$. This value provides a practical trade-off between reward maximization and capacity enforcement; Appendix~\ref{app:diff_eta} reports sensitivity analyses for other values of $\eta$. 
The main results in Figures \ref{fig:camerelyon} and \ref{fig:main_results} are averaged over 100 randomized task streams. For readability, error bars are only reported in Appendix~\ref{app:uncertainty}:  across all datasets, contextual policies, and evaluated capacities, the standard errors of the mean error rates are small relative to the corresponding improvements over the non-contextual benchmark.  

For reference, we consider a non-contextual baseline policy that randomly assigns each task to agent $a$ with fixed probability $\alpha_a$, independent of $x$, and thus does not exploit heterogeneity in $\mu_a(x)$.
This baseline is also asymptotically equivalent in expected reward to learned non-contextual policies that estimate only marginal agent accuracy, such as $a_t \in \arg\max_a \left\{\hat{\mu}_{a,t}-\eta Q_{t,a}\right\}$ (details in Appendix~\ref{app:noncontextual-equivalence}),
because such policies do not condition assignments on \(x\). 
This comparison in effect also provides a component ablation: replacing the contextual reward model with marginal reward estimates, while retaining the same queue-based capacity controller, yields the non-contextual baseline asymptotically, isolating the contribution of contextual expertise modeling.
We note that existing algorithms used in the fair-bandit literature \citep{li2019combinatorial,huang2020thompson, chen2020fair,claure2020multi} are not comparable as baselines since they address different problems than ours, e.g., where there is a common minimum rate of workload for all agents and workloads are unconstrained above that, and/or the algorithms are non-contextual  (resulting in identical performance to our non-contextual algorithm implementation). 

To assess how online and offline policy learning methods compare in terms of performance, we also consider an offline unconstrained benchmark that learns each agent’s expertise using all available test data, and revisits this test data to assign tasks via $\phi(x_t) = \arg\max_{a} \hat{\mu}_a(x_t)$. This is an optimistic estimate of the error rate when agent expertise is learned offline without capacity constraints.

\subsection{Experiments with Image, Tabular, and Text Classification Tasks}
\label{subsec:main_results}
We start by introducing an example with the \textbf{Camelyon17} medical image classification dataset \citep{bandi2018detection}. For illustration, we train 2 linear classifiers, each trained on entirely different data from 2 different hospitals, and then evaluated on a third hospital. We use the first principal component of several image features as context for the contextual policies (see Appendix \ref{app:dataset} for details). Figure~\ref{fig:camerelyon_pca} visualizes heterogeneity in agent expertise. The two agents exhibit different accuracy profiles as a function of context, and the estimated contextual reward gap $\Delta(x)$ changes across the context space. This indicates that neither agent is uniformly best and that there is potential value in context-dependent assignment.

Figure~\ref{fig:camerelyon} reports the average error rate as a function of Agent 1's capacity, where Agent 1 is trained on Hospital 1 data. The upper dotted line is the error rate for the non-contextual (random) policy with capacity $\alpha_1$ (x-axis), where the two endpoints correspond to assigning all tasks to a single agent (resulting in the marginal error rates of the two agents).
The lower dotted lines are the optimistic error rates from the offline unconstrained context-dependent models (logistic and tree-based, as described above). The intermediate lines are the error rates (averaged over runs) obtained from running the greedy and Thompson contextual policies for each of the logistic and tree-based context models, at different specific capacities $\alpha_1 \in \{0.2, 0.4, 0.5, 0.6, 0.8\}$.
All contextual policies outperform the non-contextual baseline across all capacity levels. Moreover, the combined policy can achieve a significantly lower error than the non-contextual policy (e.g., at $\alpha_1 = 0.5$), demonstrating that leveraging context-dependent expertise systematically improves performance, and can result in performance better than either agent’s marginal error rate, consistent with the theoretical results in Section~\ref{subsec:cap_constraint}. The offline unconstrained policies assign roughly half of the tasks (not shown on the graph) to Agent~1 for both logistic and tree models, indicating that the unconstrained context-dependent allocation is close to $0.5/0.5$ split. Thus, $\alpha_1 = 0.5$ is the empirical capacity allocation closest to allowing full exploitation of disagreement, in the sense of Proposition \ref{prop:full_disagreement}. The ``error gap" between the optimistic offline error of 0.16 and the online policies' average error of roughly $0.23$ at $\alpha_1 = 0.5$ reflects the cost of online learning, while moving away from this point introduces an additional error penalty from enforcing the constraints.

\begin{figure}[h]
    \centering 
   \begin{subfigure}[c]{0.46\linewidth}
        \centering
        \includegraphics[width=\linewidth]{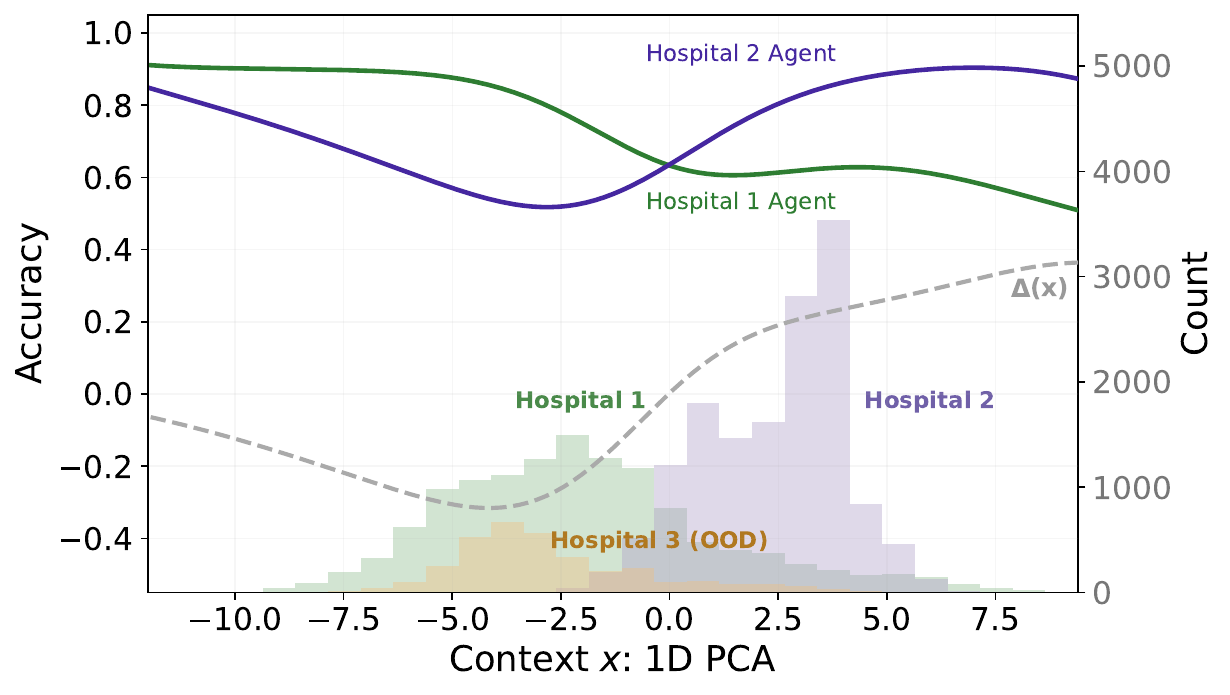}
            \caption{}
        \label{fig:camerelyon_pca}
    \end{subfigure}
    \begin{subfigure}[c]{0.48\linewidth}
        \centering
        \includegraphics[width=\linewidth]{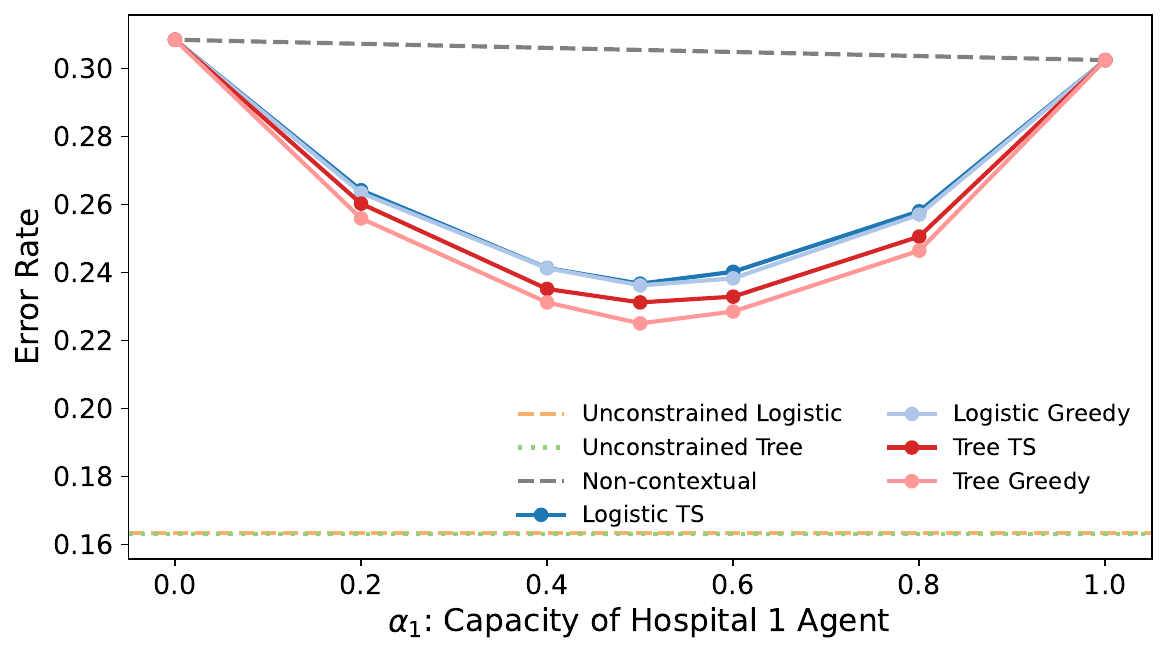}
        \caption{} 
        \label{fig:camerelyon}
    \end{subfigure}%
    \caption{(a) Empirical accuracy as a function of context, (b) error rate vs capacity of Agent 1 averaged over 2000 image classification tasks and 100 random online streams from the \textbf{Camelyon17}. Uncertainty estimates are reported in Appendix ~\ref{app:uncertainty}}.
\end{figure}

We further evaluate our framework on additional tabular, text, and image datasets, with the main text reporting one representative dataset of each type and Appendix~\ref{app:additional_datasets} reporting
the remaining four. 
For the tabular dataset, we construct two agents with heterogeneous expertise: a logistic regression model trained on a random subset of up to 20\% of the features, and an XGBoost classifier trained on the remaining features. Appendix~\ref{app:feature_split_robustness} shows that the gain over random assignment is not specific to this 20/80 split. The contextual policy observes the full feature space. Each agent is trained on 70\% of the data, and the remaining 30\% is used to simulate the online task sequence. We report results here on \textbf{Bank} (15 features, 12,357 tasks, predicting term deposit subscription). 
For the text-based setting, we evaluate pretrained LLM agents on \textbf{MMLU} \citep{hendrycks2020measuring}, using the models illustrated in Figure~\ref{fig:llm_accuracy}. 
Unlike the controlled feature-split construction used for the tabular datasets, these text-based settings use independently pretrained agents whose heterogeneous performance is not induced through feature partitioning.
For the image dataset, we use the \textbf{ImageNet16H} dataset \citep{steyvers2022bayesian}, which contains human classifications of noisy images. In this setting, the two agents are one human annotator and a fine-tuned VGG-19 convolutional neural network. 

Figure~\ref{fig:main_results} reports the average error rate as the capacity $\alpha_1$  of Agent 1 (x-axis) varies from 0 to 1. 
As in Figure~\ref{fig:camerelyon}, the endpoints correspond to the marginal error rates of the individual agents. 
Across these datasets, contextual policies (logistic and tree-based, both Thompson Sampling and Greedy) consistently outperform the non-contextual random baseline, highlighting the value of learning context-dependent expertise;  Appendix~\ref{app:additional_datasets}
shows the same pattern on four further datasets.
The combined policies do not necessarily outperform both individual agents, as performance depends on the distribution of $\Delta(x)$ and the capacity constraint (as discussed in Section \ref{subsec:cap_constraint}). Multi-agent experiments show the same pattern, with contextual assignment improving performance (Appendix~\ref{app:multi-agent}). Tree-based models generally outperform logistic models, indicating the importance of capturing nonlinear and heterogeneous patterns in $\mu_a(x)$. 
Greedy and Thompson Sampling perform similarly across these experiments, holding the contextual reward model and capacity mechanism fixed. 
The similarity in performance is likely a result of the implicit exploration induced by the capacity constraints: each positive-capacity agent continues to receive assignments and feedback, letting greedy policies learn agent expertise over time (see also \citep{bayati2020unreasonable} on the effectiveness of greedy contextual bandits).
See Appendix~\ref{app:dataset} for implementation details and \ref{app:tab_more_seeds} for additional results. Appendix \ref{app:with_free_agent} presents the case in which one agent (e.g., the AI model) has no capacity constraint. Appendix~\ref{app:capacity_validation} validates that realized
assignments closely track the target capacities across datasets.

\begin{figure}[t]
\centering

\begin{subfigure}{0.32\linewidth}
    \centering
    \includegraphics[width=\linewidth]{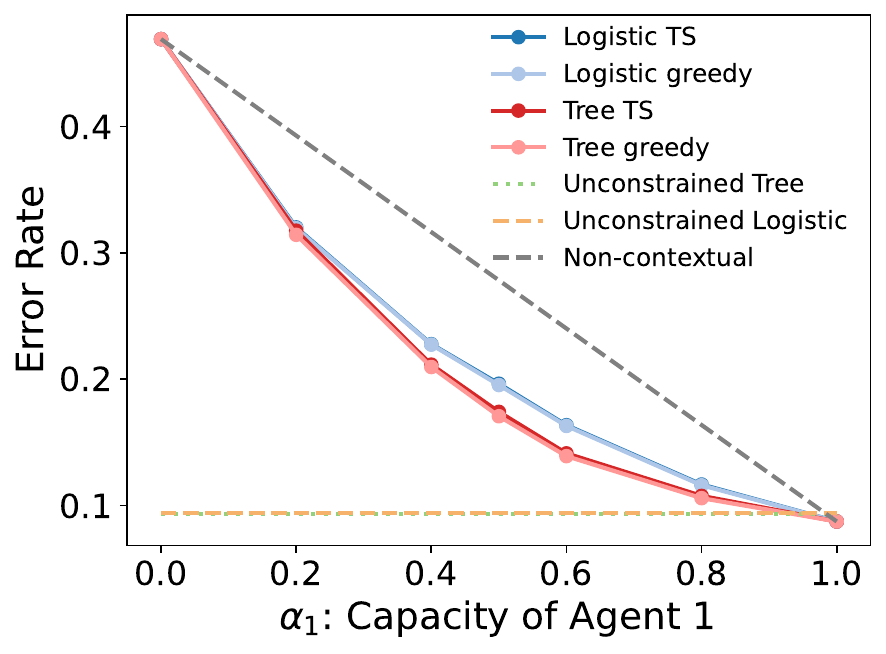}
    \caption{Bank}
\end{subfigure}
\hfill
\begin{subfigure}{0.32\linewidth}
    \centering
    \includegraphics[width=\linewidth]{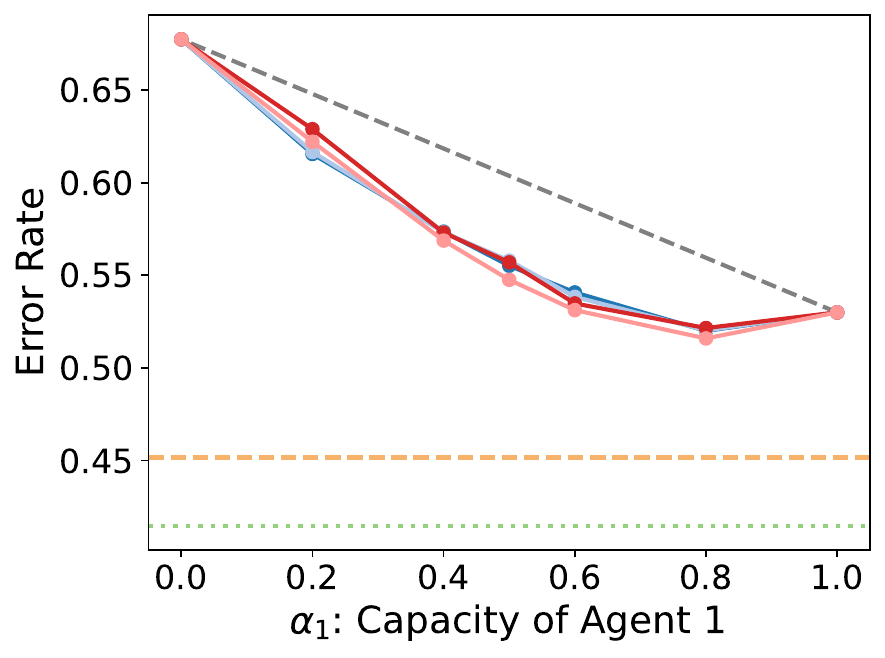}
    \caption{MMLU}
\end{subfigure}
\hfill
\begin{subfigure}{0.32\linewidth}
    \centering
    \includegraphics[width=\linewidth]{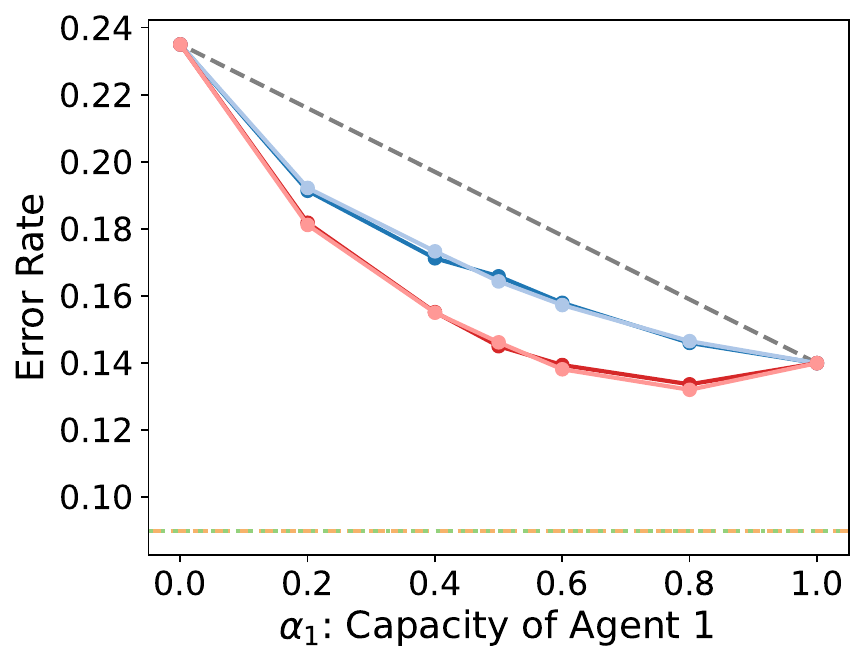}
    \caption{ImageNet16H}
\end{subfigure}

\caption{Error rate as a function of capacity (two-agent case), for policies across datasets. Points represent the average error rate over 100 runs. Uncertainty estimates are reported in Appendix~\ref{app:uncertainty}.}
\label{fig:main_results}
\end{figure}

\subsection{Mini-Batch Setting}
\label{subsec:mini_batch}
An alternative to fully online assignment, is a mini-batch setting where assignments are made for  $N_B$ tasks at a time, and outcomes are observed after the batch is completed. This setting has an inherent latency-error tradeoff: batching can improve allocation under capacity constraints by considering multiple tasks jointly, but assignments and predictions are delayed.
We illustrate this using the \textbf{Bank} dataset with the same agents as in Section~\ref{subsec:main_results}. Figure~\ref{fig:mini_batch_batch_100} compares online and mini-batch performance with batch size $N_B=100$ (50 runs) for the tree-based greedy policy. Mini-batch assignment yields modest but consistent improvements, as batching enables more efficient allocation by considering tasks jointly. We further vary the batch size from small values to the full dataset at a fixed capacity, $\alpha_1=\alpha_2=0.5$, averaging results over 10 runs for each batch size (Figure~\ref{fig:mini_batch_vary_batch_size}). The results exhibit a U-shaped error curve: small batches allow frequent updates but lead to noisy capacity enforcement, while large batches stabilize allocation but delay learning. Intermediate batch sizes balance these effects and achieve the lowest error. Overall, mini-batch assignment can improve performance over the fully online setting when a reasonable batch size is chosen, but this improvement comes at the cost of increased latency.
See Appendix~\ref{app:mini-batch} for implementation details and additional results.

\begin{figure}[t]
    \centering

    \begin{subfigure}{0.43\linewidth}
        \centering
        \includegraphics[width=0.75\linewidth]{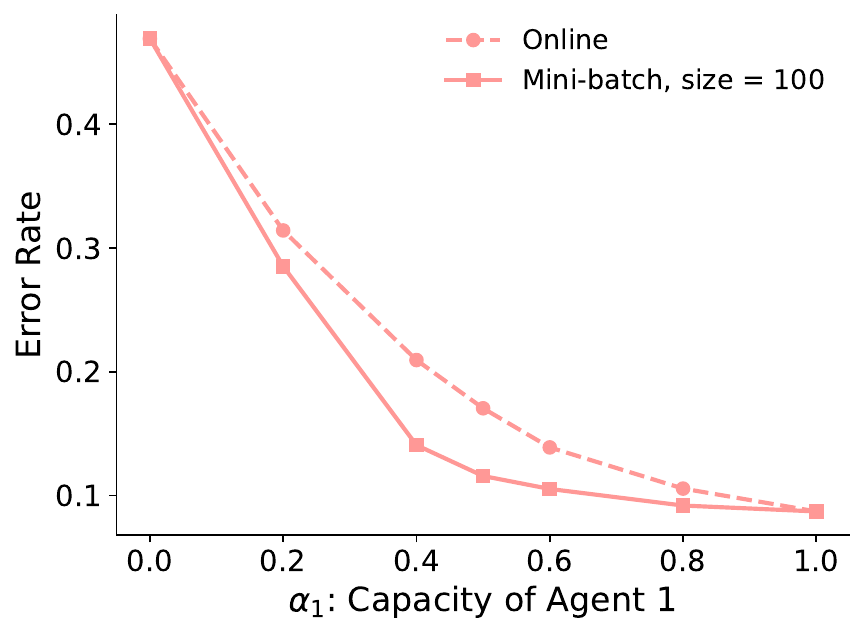}
        \caption{Online vs mini-batch assignment with batch size $N_B=100$. Error rates are averaged over 50 runs as the capacity of Agent 1 varies.}
        \label{fig:mini_batch_batch_100}
    \end{subfigure}
    \hfill
    \begin{subfigure}{0.43\linewidth}
        \centering
        \includegraphics[width=0.78\linewidth]{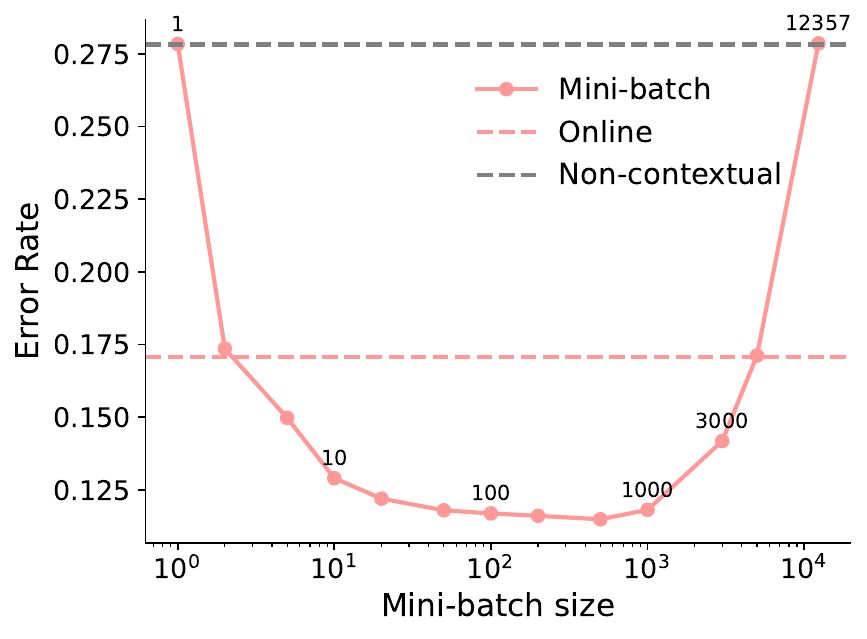}
        \caption{Effect of batch size when Agents 1 and 2 have equal capacities, $\alpha_1=\alpha_2=0.5$. Average error rates are computed over 10 runs. }
        \label{fig:mini_batch_vary_batch_size}
    \end{subfigure}

    \caption{Results of mini-batch assignment for the \textbf{Bank} dataset using the tree-based greedy policy. }
    \label{fig:mini_batch}
\end{figure}

\section{Related Work}
\label{sec:related_work}

Our work relates broadly to different strands of prior work in task routing and assignment problems.
One such line of work is {\bf learning to defer}: most work in this area studies a two-agent setting, with one model and one human, and typically learns the deferral rule offline from labeled data and observed expert decisions \citep{stronglearning, madras2018predict,mozannar2020consistent,mozannar2023should, alves2024cost}. Recent papers incorporate workload and resource constraints \citep{keswani2021towards,nguyen2025probabilistic,zhang2026fatigue,desalvo2025budgeted}, but remain largely offline. \citet{reid2024online} studied online decision deferral, but assume full-information feedback, observing outcomes for both the model and the human regardless of the routing decision. 

Another related line of work studies {\bf LLM routing}, where an algorithm learns to route tasks among models with different costs and accuracies. 
Most existing work focuses on learning routing rules from offline evaluations or preference data to balance cost and accuracy \citep{chen2023frugalgpt,ong2024routellm}, which differs from our approach, where we focus on online learning of agent abilities while respecting capacity constraints. More recent work has, however, studied online LLM routing under budget constraints \citep{panda2025adaptive, li2026wiserouter,taberner2026paretobandit}, where the constraints can be global cost budgets, so that (e.g.) all tasks could be routed to a cheap and high-performing model as long as the total budget is satisfied. This differs from our setting, where agents face long-run capacity limits that capture operational constraints such as availability, workload, or staffing. 

{\bf Fairness-constrained bandit algorithms} typically study online allocation with fairness concerns, often in non-contextual or combinatorial settings \citep{claure2020multi,li2019combinatorial,huang2020thompson}, with some contextual extensions \citep{chen2020fair}. 
Unlike \citet{chen2020fair}, who enforce a shared minimum-rate constraint across arms in expectation, we study agent-specific capacity targets enforced over realized, sequential assignments via a queue.
There is also a large literature in general on {\bf regret bounds for fair contextual bandits} \citep{li2010contextual,badanidiyuru2018bandits,balseiro2020dual,neu2022lifting,blanchet2024delay}, but this work focuses on algorithmic regret bounds rather than the characterization of assignment policies that we present here. 

Across these lines of work, expertise is learned offline, capacity constraints are absent, or the constraints are structurally different from
ours (a shared minimum-rate floor or a global cost budget). To our knowledge, no prior work studies the joint problem of online learning of context-dependent expertise under agent-specific, long-run capacity constraints.

\section{Discussion and Conclusion}
\label{sec:discussion_conclusion}
Regarding the \textbf{limitations} of our work, we evaluate our framework using proxy agents in controlled settings with relatively low-dimensional contexts. This controlled setup makes it easier to compare contextual and non-contextual assignment policies, but it may understate the complexity of real-world task allocation, where task contexts can be high-dimensional, multimodal, or only partially observed.
We also focus on classification tasks with binary reward and immediate feedback. Practical applications may involve delayed outcomes, multi-class classification with asymmetric utilities for different misclassifications, or continuous rewards, and noisier measures of agent expertise. \textbf{Future work} can extend the framework in several directions. A natural next step is to study high-dimensional and unstructured contexts, such as text or interaction histories, where the system can learn useful task features before identifying which agents are best suited for each task. Another direction is to move beyond assigning each task to a single agent and allow multiple humans or AI agents to collaborate on the same task, requiring mechanisms for aggregation, disagreement, and capacity-aware coordination. Also of interest is evaluating these policies in real human-AI workflows, where expertise, availability, and behavioral responses to assignment policies may evolve over time.

In summary, we study sequential task allocation under capacity constraints when agents have unknown, context-dependent expertise. We show theoretically and empirically that contextual assignment improves performance by learning which agents are best suited for which tasks, while capacity constraints induce implicit exploration by ensuring that all agents continue to receive assignments. Our results highlight the importance of allocation policies that learn heterogeneous expertise while respecting real-world capacity limits in human-AI collaboration.


\acks{This work was supported by National Science Foundation under award NSF 2505006, by the Hasso Plattner Institute (HPI) Research Center in Machine Learning and Data Science at UCI, and by funding support from Google and from SAP. }


\newpage

\appendix
\section{Optimal Policy Without Capacity Constraint}
\label{app:general_unconstrained}

Here, we characterize the optimal policy for $A \ge 2$ agents and compare it with non-contextual assignment.

The objective is
\[
\max_{\phi} \; 
\mu^{\phi}
=
\mathbb{E}_{x}\!\left[
\sum_{a=1}^A P_\phi(a \mid x)\mu_a(x)
\right],
\]
where $P_\phi(a \mid x)$ denotes the probability that policy $\phi$ assigns context $x$ to agent $a$, with $\sum_a P_\phi(a \mid x)=1$ for all $x$.

Since the objective is an expectation over $x$, the assignment decision at one context does not affect assignments at other contexts. Thus, for each fixed $x$, the policy solves
\[
\max_{P_\phi(\cdot \mid x)}
\sum_{a=1}^A P_\phi(a \mid x)\mu_a(x),
\quad
\text{s.t.} \quad
\sum_a P_\phi(a \mid x)=1.
\]
Because this objective is linear in $P_\phi(\cdot \mid x)$, an optimum is attained by assigning all probability to an agent with the largest conditional expected reward. Therefore, the optimal policy can be written as
\[
\phi^\star(x)
=
\arg\max_{a \in \{1,\dots,A\}} \mu_a(x).
\]
The resulting marginal reward is
\[
\mu^{\phi^\star}
=
\mathbb{E}_{x}\!\left[
\max_{a \in \{1,\dots,A\}} \mu_a(x)
\right].
\]

We now compare this policy with a non-contextual random assignment policy that assigns each task to agent $a$ with fixed probability $P_a$, independent of $x$, where $\sum_a P_a=1$. Its marginal reward is
\[
\mu^{\mathrm{rand}}
=
\sum_{a=1}^A P_a \mu_a
=
\mathbb{E}_{x}\!\left[
\sum_{a=1}^A P_a \mu_a(x)
\right].
\]
For every context $x$, $\max_a \mu_a(x)\ge
\sum_{a=1}^A P_a \mu_a(x)$,
since the maximum is at least any weighted average. Taking expectations over $x$ gives
\[
\mu^{\phi^\star}
\ge
\mu^{\mathrm{rand}}.
\]

The inequality is strict when the best agent varies across contexts, and the non-contextual policy sometimes assigns tasks to lower-performing agents. Intuitively, contextual assignment improves by routing each context to the agent with the highest expected reward.

A best fixed-agent policy is a special case of non-contextual assignment. It chooses
\[
a^{\max} = \arg\max_{a \in \{1,\dots,A\}} \mu_a
\]
and assigns every task to $a^{\max}$, achieving reward $\mu^{\max} = \max_a \mu_a$.

Since $\max_a \mu_a(x) \ge \mu_{a^{\max}}(x)$ for every $x$, taking expectations gives
\[
\mu^{\phi^\star}
=
\mathbb{E}_x[\max_a \mu_a(x)]
\ge
\mathbb{E}_x[\mu_{a^{\max}}(x)]
=
\mu^{\max}.
\]
Thus, the optimal contextual policy weakly dominates the best fixed-agent policy, with strict improvement whenever some other agent outperforms $a^{\max}$ on a positive-measure set of contexts.

\section{Analysis of  the Two-Agent Case Without Capacity Constraint}
\label{app:two_agent}
\subsection{Optimal Policy}
\label{app:two_agent_optimal_policy_no_constained}
Here, we provide a formal characterization of the two-agent case and quantify the performance gap between contextual and non-contextual policies. Consider two agents with conditional expected rewards $\mu_1(x)$ and $\mu_2(x)$, and define the contextual reward gap
\[
\Delta(x) := \mu_1(x) - \mu_2(x).
\]

The optimal (unconstrained) policy assigns each context to the agent with the higher conditional expected reward,
\[
    \phi^\star(x) = \arg\max_{a \in \{1,2\}} \mu_a(x),
\]
which is equivalently characterized by selecting agent $1$ when $\Delta(x) \ge 0$ and agent $2$ otherwise. This induces decision regions
\[
    \mathcal{R}_1^\star := \{x : \Delta(x) \ge 0\}, 
    \qquad
    \mathcal{R}_2^\star := \{x : \Delta(x) < 0\}.
\]
The marginal reward of the optimal policy is
\[
    \mu^{\phi^\star}
    = \int_{x \in \mathcal{R}_1^\star} \mu_1(x)\, P_X(x)\, dx
    + \int_{x \in \mathcal{R}_2^\star} \mu_2(x)\, P_X(x)\, dx.
\]

\subsection{Optimal Policy Weakly Dominates Non-contextual Policies}
\label{app:two_agent_optimal_over_noncontextual_unconstrained}

Here we compare the optimal contextual policy with the best context-independent policy. Consider the policy that ignores $x$ and always assigns tasks to the agent with the highest marginal reward $ a^{\max} = \arg\max_{a \in \{1,2\}} \mu_a$.
Without loss of generality, suppose $a^{\max}=1$. Then this policy achieves
\[
\mu^{\max}
=
\int \mu_1(x)\,dP_X(x)
=
\mu_1.
\]

The optimal contextual policy weakly dominates this baseline. Its performance gain is
\begin{align*}
\mu^{\phi^\star} - \mu^{\max}
    &= \int_{x \in \mathcal{R}_2^\star} \bigl(\mu_2(x) - \mu_1(x)\bigr)\, P_X(x)\, dx = \int_{x \in \mathcal{R}_2^\star} -\Delta(x)\, P_X(x)\, dx\ge\; 0.
\end{align*}
The inequality holds because on $\mathcal{R}_2^\star=\{x:\Delta(x)<0\}$ (i.e.,  agent $2$ outperforms agent $1$). Thus, contextual assignment improves over the best fixed-agent policy precisely by switching to agent $2$ on these contexts.

\subsection{Outperforming Individual Agents}
\label{app:two_agent_unconstrained_outperform}

Here, we show that when two agents have complementary expertise across contexts, the combined contextual policy can outperform either individual agent, i.e.,
\[
P_X(\Delta(x)>0)>0
\quad \text{and} \quad
P_X(\Delta(x)<0)>0.
\]
Then the optimal contextual policy strictly outperforms both fixed-agent policies.

From Appendix~\ref{app:two_agent_optimal_over_noncontextual_unconstrained}, if $\mu_1 \ge \mu_2$, the gain over always assigning to agent $1$ is
\[
\mu^{\phi^\star}-\mu_1
=
\int_{\Delta(x)<0}
\bigl(\mu_2(x)-\mu_1(x)\bigr)\,dP_X(x)
>0,
\]
since $\mu_2(x)>\mu_1(x)$ on $\{\Delta(x)<0\}$ and this region has positive probability. Similarly, the gain over always assigning to agent $2$ is
\[
\mu^{\phi^\star}-\mu_2
=
\int_{\Delta(x)\ge 0}
\bigl(\mu_1(x)-\mu_2(x)\bigr)\,dP_X(x)
>0,
\]
since $\mu_1(x)>\mu_2(x)$ on $\{\Delta(x)>0\}$ and this region has positive probability. Therefore,
\[
\mu^{\phi^\star} > \max\{\mu_1,\mu_2\}.
\]
Thus, when each agent has a positive-measure region of relative advantage, the contextual policy can outperform any single agent by routing each context to the better-performing agent. This is conceptually related to dynamic classifier selection, where classifier complementarity across regions of the feature space enables instance-wise selection to improve over individual classifiers \citep{giacinto2000dynamic}.

\subsection{Role of Disagreement and Magnitude of Improvement}
\label{app:two_agent_unconstrained_disagreement}

We now express the improvement from contextual assignment in terms of the disagreement regions. Define
\[
\mathcal{D}_1 := \{x:\Delta(x)>0\},
\qquad
\mathcal{D}_2 := \{x:\Delta(x)<0\},
\]
with
\[
p_1 := P_X(\mathcal{D}_1),
\qquad
p_2 := P_X(\mathcal{D}_2).
\]
On $\mathcal{D}_1$, agent $1$ outperforms agent $2$, while on $\mathcal{D}_2$, agent $2$ outperforms agent $1$. The contextual policy exploits all such disagreement by assigning $\mathcal{D}_1$ to agent $1$ and $\mathcal{D}_2$ to agent $2$.

Relative to always assigning to agent $1$, the improvement is
\[
\mu^{\phi^\star}-\mu_1
=
\int_{\mathcal{D}_2} \bigl(\mu_2(x)-\mu_1(x)\bigr)\,dP_X(x)
=
\int_{\mathcal{D}_2} -\Delta(x)\,dP_X(x).
\]
Similarly, relative to always assigning to agent $2$, the improvement is
\[
\mu^{\phi^\star}-\mu_2
=
\int_{\mathcal{D}_1} \bigl(\mu_1(x)-\mu_2(x)\bigr)\,dP_X(x)
=
\int_{\mathcal{D}_1} \Delta(x)\,dP_X(x).
\]

Equivalently, these improvements can be written as
\[
\mu^{\phi^\star}-\mu_1
=
p_2 \cdot \mathbb{E}\bigl[-\Delta(X)\mid X\in\mathcal{D}_2\bigr],
\]
and
\[
\mu^{\phi^\star}-\mu_2
=
p_1 \cdot \mathbb{E}\bigl[\Delta(X)\mid X\in\mathcal{D}_1\bigr].
\]
Thus, relative to a fixed-agent assignment, the improvement from contextual assignment depends both on the probability mass of contexts in which the fixed agent is worse than the alternative agent and on the magnitude of the contextual reward gap within those contexts.

\section{Optimal Policy under Capacity Constraints}
\label{app:general_case}

Here, we show the optimal policy for $A \ge 2$ agents under capacity constraints and establish its performance relative to non-contextual assignment.

Our goal is
\[
\max_{\phi} \; \mathbb{E}_{x}\!\left[\sum_{a=1}^A P_\phi(a \mid x)\mu_a(x)\right]
\quad \text{s.t.} \quad
\mathbb{E}_{x}[P_\phi(a \mid x)] = \alpha_a, \ \forall a,
\]
where $\sum_{a=1}^A \alpha_a = 1$ and $P_\phi(a \mid x)$ denotes the probability that policy $\phi$ assigns context $x$ to agent $a$, with $\sum_a P_\phi(a \mid x)=1$ for all $x$.

Introduce Lagrange multipliers $\lambda_1,\dots,\lambda_A$ and consider the Lagrangian

\begin{align*}
\mathcal{L}(\phi,\lambda)
&=
\mathbb{E}_{x}\!\left[\sum_{a=1}^A P_\phi(a \mid x)\mu_a(x)\right]
    -
    \sum_{a=1}^A \lambda_a
    \left(
        \mathbb{E}_{x}[P_\phi(a \mid x)] - \alpha_a
    \right) \\
&=
\mathbb{E}_{x}\!\left[\sum_{a=1}^A P_\phi(a \mid x)\mu_a(x)\right]
    -
    \sum_{a=1}^A \lambda_a \mathbb{E}_{x}[P_\phi(a \mid x)]
    +
    \sum_{a=1}^A \lambda_a \alpha_a \\
&=
\mathbb{E}_{x}\!\left[
        \sum_{a=1}^A P_\phi(a \mid x)\bigl(\mu_a(x)-\lambda_a\bigr)
    \right]
    +
    \sum_{a=1}^A \lambda_a \alpha_a.
\end{align*}

For fixed $\lambda$, the second term is constant in $\phi$. Since the objective is an expectation over $x$, the assignment decision at one context does not affect others. Thus, we can determine the optimal assignment separately for each $x$:
\[
\max_{P_\phi(\cdot \mid x)} \sum_{a=1}^A P_\phi(a \mid x)\bigl(\mu_a(x)-\lambda_a\bigr),
\quad \text{s.t.} \ \sum_a P_\phi(a \mid x)=1,
\]
since this objective is linear in $P_\phi(\cdot \mid x)$, an optimum is attained by assigning all probability to an agent with the largest adjusted reward; that is, $P_\phi(a \mid x) \in \{0,1\}$ for all $a$. Therefore, the optimal policy can be written as a deterministic rule:
\[
\phi^\star(x) = \arg\max_{a \in \{1,\dots,A\}} \bigl\{\mu_a(x)-\lambda_a\bigr\}.
\]

There exist constants $\{\lambda_a\}$ such that the resulting policy satisfies the constraints
\[
\mathbb{E}_{x}[P_{\phi^\star}(a \mid x)] = \alpha_a.
\]
Under the continuous-context assumption stated in Section~\ref{subsec:cap_constraint}, this policy can be chosen deterministically while satisfying the capacity constraints exactly.

To compare with the non-contextual assignment, define $g_a(x) := \mu_a(x)-\lambda_a$. Then the optimal policy is
\[
\mu^{\phi^\star}
=
\mathbb{E}_{x}\!\left[\max_a g_a(x)\right]
+
\sum_{a=1}^A \lambda_a \alpha_a,
\]
while the non-contextual policy yields
\[
\begin{aligned}
\mu^{\mathrm{rand}}
&= \sum_{a=1}^A \alpha_a \mu_a  \\
&= \sum_{a=1}^A \alpha_a \mathbb{E}_x[\mu_a(x)] \\
&= \mathbb{E}_x\!\left[\sum_{a=1}^A \alpha_a \mu_a(x)\right] \\
&= \mathbb{E}_x\!\left[\sum_{a=1}^A \alpha_a \bigl(g_a(x)+\lambda_a\bigr)\right] \\
&= \mathbb{E}_x\!\left[\sum_{a=1}^A \alpha_a g_a(x)\right]
+ \sum_{a=1}^A \lambda_a \alpha_a .
\end{aligned}
\]

Since for any numbers $\{z_a\}$ and weights $\{\alpha_a\}$ with $\sum_a \alpha_a=1$,
\[
\max_a z_a \ge \sum_{a=1}^A \alpha_a z_a,
\]
we have pointwise
\[
    \max_a g_a(x)\ge \sum_{a=1}^A \alpha_a g_a(x).
\]
Taking expectations yields $\mu^{\phi^\star} \ge \mu^{\mathrm{rand}}$.

The inequality is strict whenever the best agent varies across contexts, so that the non-contextual policy averages over suboptimal agents while $\phi^\star$ selects the best agent at each context.

Finally, the constants $\lambda_a$ can be interpreted as shadow prices associated with the capacity constraints. They quantify the marginal value of relaxing each constraint and characterize how the optimal policy trades off reward and capacity across agents.

\section{Two-Agent Case under Capacity Constraints}
\label{app:two_agent_capacity}
Throughout this section, we distinguish two cases. We refer to the \emph{non-dominant case} as the setting where two agents exhibit heterogeneous expertise across contexts, i.e., $P_X(\Delta(x)>0)>0$ and $P_X(\Delta(x)<0)>0$. We refer to the \emph{dominant case} as the setting where one agent weakly dominates the other, i.e., $\mu_1(x)\ge \mu_2(x)$ for all $x$.

\subsection{Optimal Policy under Capacity Constraints}
\label{app:two_agent_optimal_policy}
Here we show that the optimal constrained policy assigns agent $1$ to the $\alpha$-mass of contexts with the largest values of $\Delta(x)$, where $\Delta(x)= \mu_1(x)-\mu_2(x)$. 

Let the target workload capacities be $\alpha_1=\alpha$ and $\alpha_2=1-\alpha$. Any policy $\phi$ induces regions
\[
    \mathcal{R}_1 := \{x \in \mathcal{X} : \phi(x)=1\},
    \qquad
    \mathcal{R}_2 := \mathcal{X}\setminus \mathcal{R}_1,
\]
with
\[
    \int_{x \in \mathcal{R}_1} P_X(x)\,dx = \alpha.
\]

Any feasible policy induces a region $\mathcal{R}_1$ with $P_X(\mathcal{R}_1)=\alpha$, assigning agent $1$ to $\mathcal{R}_1$ and agent $2$ to its complement.

The expected reward is
\begin{align*}
    \mu^{\phi}
    &= \int_{x \in \mathcal{R}_1} \mu_1(x) P_X(x)\,dx
     + \int_{x \in \mathcal{R}_2} \mu_2(x) P_X(x)\,dx \\
    &= \mu_2 + \int_{x \in \mathcal{R}_1} \Delta(x) P_X(x)\,dx,
\end{align*}

Thus, the problem reduces to selecting a set $\mathcal{R}_1$ of mass $\alpha$ that maximizes
\[
\int_{\mathcal{R}_1} \Delta(x)\,dP_X(x).
\]

This is solved by choosing the $\alpha$-mass of contexts with the largest values of $\Delta(x)$. Therefore, there exists a threshold $\tau_\alpha$ such that
\[
\mathcal{R}_1^\star = \{x : \Delta(x)\ge \tau_\alpha\}, 
\qquad
\mathcal{R}_2^\star = \mathcal{X}\setminus \mathcal{R}_1^\star.
\]

This threshold characterization is consistent with the general Lagrangian solution in Appendix~\ref{app:general_case}. There, the optimal policy takes the form
\[
\phi^\star(x) = \arg\max_{a \in \{1,2\}} \{\mu_a(x) - \lambda_a\}.
\]
For two agents, this assigns to agent $1$ whenever
\[
\mu_1(x) - \lambda_1 \ge \mu_2(x) - \lambda_2,
\]
which can be rewritten as
\[
\Delta(x) := \mu_1(x) - \mu_2(x) \ge \lambda_1 - \lambda_2.
\]
Thus, the threshold $\tau_\alpha$ corresponds to the shadow-price difference $\lambda_1 - \lambda_2$, determined by the capacity constraint.

The optimal policy therefore assigns agent $1$ to the top $\alpha$-mass of contexts ranked by $\Delta(x)$, and agent $2$ to the remainder. The resulting marginal reward is
\[
\mu^{\phi^\star}
= \mu_2 + \int_{x \in \mathcal{R}_1^\star} \Delta(x) P_X(x)\,dx.
\]

\subsection{Constrained Optimal Policy Weakly Dominates Non-contextual Policy}
\label{app:two_agent_optimal_over_radnom}
Here we establish that this policy weakly dominates the random policy in both the \textit{non-dominant case} and the \textit{dominant case}.
As shown in Appendix \ref{app:two_agent_optimal_policy}, the constrained optimal policy has reward 
\[
    \mu^{\phi^\star}
    = \mu_2 + \int_{x \in \mathcal{R}_1^\star} \Delta(x) P_X(x)\,dx.
\]
By contrast, the non-contextual (i.e., random) policy has reward
\[
    \mu^{\mathrm{rand}}
    = \alpha \mu_1 + (1-\alpha)\mu_2
    = \mu_2 + \alpha \int \Delta(x) P_X(x)\,dx.
\]
Then, the degree to which the optimal policy improves over the random policy is 
\begin{align*}
    \mu^{\phi^\star} - \mu^{\mathrm{rand}}
    &= \int_{x \in \mathcal{R}_1^\star} \Delta(x) P_X(x)\,dx
     - \alpha \int \Delta(x) P_X(x)\,dx \\
    &= (1-\alpha)\int_{x \in \mathcal{R}_1^\star} \Delta(x) P_X(x)\,dx
     - \alpha \int_{x \in \mathcal{R}_2^\star} \Delta(x) P_X(x)\,dx
    \;\ge\; 0.
\end{align*}

\subsection{Outperforming Individual Agents}
\label{app:two_agent_outperform}
Here, we show that, for the \textit{non-dominant case}, the combined policy can outperform either individual agent when performance varies across contexts.

Suppose
\[
P_X(\Delta(x)>0)>0
\quad \text{and} \quad
P_X(\Delta(x)<0)>0.
\]

Let $\alpha = P_X(\Delta(x)\ge 0)$. Then the optimal set becomes 
\[
\mathcal{R}_1^\star = \{x:\Delta(x)\ge 0\},
\]
so each context is assigned to the better-performing agent.

The resulting policy assigns each context $x$ to the better-performing agent. Its reward becomes:
\begin{align*}
\mu^{\phi^\star}
& =
\int_{\Delta(x)\ge 0} \mu_1(x)\,dP_X(x)
+
\int_{\Delta(x)<0} \mu_2(x)\,dP_X(x)\\
& =\mu_1
+
\int_{\Delta(x)<0} \bigl(\mu_2(x)-\mu_1(x)\bigr)\,dP_X(x) \\
&>
\mu_1,
\end{align*}

since $\mu_2(x)>\mu_1(x)$ on $\{\Delta(x)<0\}$.

Similarly: 
\[
    \mu^{\phi^\star}
    = \mu_2 + \int_{\Delta(x)\ge 0} \bigl(\mu_1(x)-\mu_2(x)\bigr) P_X(x)\,dx
    > \mu_2.
\]

so
\[
\mu^{\phi^\star} > \max\{\mu_1,\mu_2\}.
\]

\subsection{Role of Disagreement and Magnitude of Improvement}
\label{app:two_agent_disagreement}

Define the disagreement region between two agents
\[
\mathcal{D} := \{x \in \mathcal{X} : \Delta(x) \neq 0\},
\]
and partition it into
\[
\mathcal{D}_1 := \{x : \Delta(x) > 0\}, \qquad
\mathcal{D}_2 := \{x : \Delta(x) < 0\}.
\]
Let
\[
p_1 := P_X(\mathcal{D}_1), \quad
p_2 := P_X(\mathcal{D}_2), \quad
t := P_X(\mathcal{D}) = p_1 + p_2.
\]

On $\mathcal{D}_1$, agent 1 strictly outperforms agent 2, while the reverse holds on $\mathcal{D}_2$. The optimal policy exploits this heterogeneity by allocating capacity to contexts with the largest values of $\Delta(x)$, whereas the non-contextual random policy ignores this structure. As shown in Appendix \ref{app:two_agent_optimal_over_radnom}, the improvement over the non-contextual policy is
\[
\mu^{\phi^\star} - \mu^{\mathrm{rand}}
=
\int_{\mathcal{R}_1^\star} \Delta(x)\,dP_X(x)
-
\alpha \int \Delta(x)\,dP_X(x)
\;\ge\; 0.
\]

When capacity allows full exploitation of disagreement, i.e., $p_1 \le \alpha \le 1-p_2$, the policy assigns $\mathcal{D}_1$ to agent $1$ and $\mathcal{D}_2$ to agent $2$.

In this case, the improvement over the non-contextual random policy can be written as\[\mu^{\phi^\star} - \mu^{\mathrm{rand}}=(1-\alpha)\int_{\mathcal{D}_1}\Delta(x)\,dP_X(x)+\alpha\int_{\mathcal{D}_2}-\Delta(x)\,dP_X(x).\]Equivalently,\[\mu^{\phi^\star} - \mu^{\mathrm{rand}}=(1-\alpha)p_1\cdot\mathbb{E}\bigl[\Delta(X)\mid X\in\mathcal{D}_1\bigr]+\alpha p_2\cdot\mathbb{E}\bigl[-\Delta(X)\mid X\in\mathcal{D}_2\bigr].\]Thus, the magnitude of improvement is determined by both the size of the disagreement regions and the average contextual reward gap within those regions, weighted by the capacity $\alpha$.

\subsection{Equivalence of Learned and Random Non-Contextual Policies}
\label{app:noncontextual-equivalence}

Consider any non-contextual policy whose assignment decision does not depend on the current context $x_t$. This includes policies that learn only marginal agent performance, for example, policies that maintain estimates $\hat{\mu}_{a,t}$ of $\mu_a$ and choose
\[
a_t = \arg\max_a \left\{\hat{\mu}_{a,t}-\eta Q_{t,a}\right\}.
\]

Let $H_t$ denote the history before the reward at round $t$. Because the policy is non-contextual, conditioning on $a_t=a$ does not change the distribution of $x_t$. Therefore, 
\[
\mathbb{E}[r_{t,a} \mid H_t, a_t=a]
=
\mathbb{E}_{x_t}[\mathbb{E}[r_{t,a}\mid x_t]]
=
\mathbb{E}_{x_t}[\mu_a(x_t)]
=
\mu_a.
\]
So we have,
\[
\mathbb{E}[r_{t,a_t} \mid H_t]
=
\sum_a \Pr(a_t=a \mid H_t)\mu_a.
\]
Averaging over time gives
\[
\frac{1}{T}\sum_{t=1}^T \mathbb{E}[r_{t,a_t}]
=
\sum_a
\left(
\frac{1}{T}\sum_{t=1}^T \Pr(a_t=a)
\right)\mu_a.
\]
As the capacity mechanism enforces the long-run assignment capacities
\[
\frac{1}{T}\sum_{t=1}^T \Pr(a_t=a) \to \alpha_a,
\]
then the limiting expected reward of the learned non-contextual policy is
\[
\sum_a \alpha_a \mu_a.
\]

The random non-contextual baseline assigns agent $a$ with probability $\alpha_a$, independently of $x_t$, and therefore has the same expected reward:
\[
\mathbb{E}[r_{t,a_t}]
=
\sum_a \alpha_a \mathbb{E}_{x_t}[\mu_a(x_t)]
=
\sum_a \alpha_a \mu_a.
\]
Thus, under the same long-run assignment capacities, learned and random non-contextual policies are asymptotically equivalent in expected reward. Finite-sample realizations may differ due to initialization, queue transients, tie-breaking, and reward noise, but without conditioning on $x_t$, learned marginal accuracies cannot improve performance through context-dependent sorting.

\section{Contextual Reward Model Implementation Details}
\label{app:context_model_para_update}
This section provides implementation details for the contextual reward model updates used in the online assignment procedure described in Section~\ref{subsec:virtual_quque}.

\begin{algorithm}[htbp]
\caption{Online capacity-constrained assignment}
\label{alg:online_assignment}
\begin{algorithmic}[1]
\STATE \textbf{Input:} Capacities $\{\alpha_a\}$, queue penalty $\eta$
\STATE Initialize estimates $\mu_{a,1}(\cdot)$ and queues $Q_{1,a}=0$ for all $a$
\FOR{$t=1,2,\dots$}
    \STATE Observe context $x_t$
    \STATE Select agent $a_t = \arg\max_a \bigl\{ \mu_{a,t}(x_t) - \eta Q_{t,a} \bigr\}$
    \STATE Observe reward $r_{t,a_t}$ and update context models and posterior distributions 
    \STATE Update queues $Q_{t+1,a} = \bigl[ Q_{t,a} + \mathbbm{1}\{a_t=a\} - \alpha_a \bigr]_+$
\ENDFOR
\end{algorithmic}
\end{algorithm} 

\subsection{Logistic Contextual Model: Greedy and Thompson Sampling}
\label{app:logistics_posterior_update}

Here we describe the posterior approximation and online update procedure used for Thompson sampling in logistic models. For each agent $a$, rewards are modeled as
\[
r_{t,a} \sim \mathrm{Bernoulli}(\sigma(\theta_a^\top x_t)),
\]
where $\sigma(z) = (1+e^{-z})^{-1}$.

For each agent $a$, let $\mathcal{D}_{a,t}=\{(x_s,r_{j,a}) : s \le t,\ a_j=a\}$ denote the bandit feedback observed for that agent. Given $\mathcal{D}_{a,t}$, the posterior over $\theta_a$ is intractable. Following \citet{chapelle2011empirical}, we use a Laplace approximation centered at the maximum a posteriori (MAP) estimate. The posterior is approximated as
\[
\theta_a \mid \mathcal{D}_{a,t} \approx \mathcal{N}(\hat{\theta}_{a,t}, \Sigma_{a,t}),
\]
where $\hat{\theta}_{a,t}$ is the MAP estimate and $\Sigma_{a,t}$ is the inverse Hessian of the negative log-posterior.

To enable efficient online updates, we maintain $\hat{\theta}_{a,t}$ and $\Sigma_{a,t}$ incrementally. For a new observation $(x_t, r_{t,a})$, define
\[
\hat{p}_t = \sigma(\hat{\theta}_{a,t}^\top x_t), \quad w_t = \hat{p}_t(1 - \hat{p}_t).
\]
The covariance is updated using a Sherman--Morrison step:
\[
\Sigma_{a,t+1}
=
\Sigma_{a,t}
-
\frac{w_t \Sigma_{a,t} x_t x_t^\top \Sigma_{a,t}}{1 + w_t x_t^\top \Sigma_{a,t} x_t},
\]
and the mean is updated as
\[
\hat{\theta}_{a,t+1}
=
\hat{\theta}_{a,t}
+
\Sigma_{a,t+1}(r_{t,a} - \hat{p}_t)x_t.
\]

At each round, Thompson sampling draws
\[
\tilde{\theta}_{a,t} \sim \mathcal{N}(\hat{\theta}_{a,t}, \kappa^2 \Sigma_{a,t}),
\]
where $\kappa > 0$ controls exploration, and finally uses
\[
\mu_{a,t}(x_t) = \sigma(\tilde{\theta}_{a,t}^\top x_t)
\]
for decision making.

In our implementation, we initialize the posterior mean at zero, $\hat{\theta}_{a,1}=0$, and the covariance as $\Sigma_{a,1} = \gamma_{\mathrm{prior}}^{-1} I$.
We set $\gamma_{\mathrm{prior}}=1$ and use exploration scale $\kappa=0.5$ for Thompson sampling. For numerical stability, the covariance matrix is symmetrized after each update, and $w_t$ is clipped to the interval $[10^{-4},0.25]$.

For the greedy logistic policy, we use the same online MAP update for $\hat{\theta}_{a,t}$, but do not sample from the approximate posterior. Instead, the policy uses the plug-in estimate $\mu_{a,t}(x_t)=\sigma(\hat{\theta}_{a,t}^\top x_t)$ for decision making.

\subsection{Tree-Based Contextual Model: Greedy and Bootstrap Thompson Sampling}
\label{app:tree_contextual_model}

Here we provide the implementation details of the tree-based contextual reward model used in our online policies. For each agent $a$, we maintain an ensemble of regression trees that predicts the realized reward $r_{t,a}$ from the context $x_t$. The model is updated online using only the observations assigned to that agent.

For each agent, we use an ensemble of $B=20$ decision-tree regressors, each with maximum depth $3$ and minimum leaf size $10$. The model stores all observed pairs $(x_t,r_{t,a})$ for that agent. To reduce computational cost, the ensemble is refit every 20 observed updates rather than after every round. At each refit, each tree is trained on a bootstrap sample of the agent's observed data.

For the greedy tree policy, the estimated reward is the average prediction across trees: $\mu_{a,t}(x_t)=\frac{1}{B}\sum_{b=1}^{B} f^{(b)}_{a,t}(x_t)$,
where $f^{(b)}_{a,t}$ denotes the $b$-th tree for agent $a$ at time $t$. Before the ensemble is first fitted, we use a prior mean of $0.5$.

For Thompson sampling, we approximate posterior sampling using bootstrap randomness. At each round $t$, we sample one tree uniformly at random from the ensemble and use its prediction as the sampled reward estimate: $\tilde{\mu}_{a,t}(x_t)=f^{(b)}_{a,t}(x_t), b \sim \mathrm{Uniform}\{1,\dots,B\}.$. If the ensemble has not yet been fitted, we sample an initial reward estimate uniformly from $[0,1]$. All tree predictions are clipped to lie in $[0,1]$.

\section{Capacity Constraints and Queue-Based Implementation}
\label{app:capacity_constraints}

Here, we provide a detailed description of the capacity-constrained assignment problem and the queue-based method used to enforce long-run constraints. Let $\alpha_a \in [0,1]$ denote the target long-run fraction of tasks assigned to agent $a$, with $\sum_{a=1}^A \alpha_a = 1$. For each round $t$, define the assignment indicator
$I_{t,a} := \mathbbm{1}\{a_t=a\}$, which equals one if task $t$ is assigned to agent $a$. Following the long-term time-average constraint framework of \citet{neely2010stochastic}, the policy is required to satisfy
\[
\limsup_{T \to \infty} \frac{1}{T}\sum_{t=1}^T \mathbb{E}[I_{t,a}] \le \alpha_a,
\qquad a=1,\dots,A.
\]
When the full workload is allocated each round and $\sum_{a=1}^A \alpha_a = 1$, these upper-bound constraints bind at optimum, so the long-run expected fraction of tasks handled by agent $a$ matches $\alpha_a$.

To enforce these constraints in an online setting, we introduce a virtual queue $Q_{t,a}$ (notation first proposed by \citep{neely2010stochastic}) for each agent, which tracks the cumulative deviation between realized assignments and the target capacity. The queue evolves as
\[
Q_{t+1,a} = \bigl[ Q_{t,a} + I_{t,a} - \alpha_a \bigr]_+, \quad a= 1, .., A,
\]
where $[\cdot]_+ = \max(\cdot,0)$. $Q_{t,a}$ increases when agent $a$ is over-assigned relative to its target capacity, and remains small when assignments are balanced.

At each round $t$, the agent is selected by trading off predicted reward and queue pressure:
\[
\phi(x_t) = \arg\max_{a \in \{1,\dots,A\}} \left\{ \mu_{a,t}(x_t) - \eta Q_{t,a} \right\},
\]
where $\mu_a(x_t)$ is the agent $a$'s expected reward on context $x_t$ and $\eta \ge 0$ controls the trade-off between reward maximization and constraint satisfaction. Agents with larger queues are penalized, which encourages the policy to rebalance assignments over time while still prioritizing high-reward decisions. In this sense, the queue lengths act as dynamic shadow prices for capacity constraints.

\section{Regret Bounds for Penalized Contextual Thompson Sampling}
\label{app:regret}
Throughout, rewards lie in $[0,1]$, the estimates satisfy $\mu_{a,t}(\cdot)\in[0,1]$
(logistic outputs and clipped tree outputs), capacities satisfy $\sum_a\alpha_a=1$, and
$Q_{1,a}=0$ as in Algorithm~\ref{alg:online_assignment}. $I_{t,a}=\mathbbm 1\{a_t=a\}$ and
$[\cdot]_+$ are as in \eqref{eq:virtual_queue}, and the selection rule is
\eqref{eq:agent_select_w_cap} with a fixed tie-breaking rule. $\mathcal F_{t-1}$ denotes
the history $(x_s,a_s,r_{s,a_s})_{s<t}$, so $Q_t=(Q_{t,a})_a$ is
$\mathcal F_{t-1}$-measurable. Let $\phi^\star$ be the optimal constrained policy of
Theorem~\ref{thm:optimal_constrained} and $\mu^\star:=\mu^{\phi^\star}$ its value.
We separate two quantities:
\begin{align*}
R_T &:= \mathbb E\Big[\sum_{t=1}^T\big(\mu^\star-\mu_{a_t}(x_t)\big)\Big],\\
\tilde R_T &:= \mathbb E\Big[\sum_{t=1}^T\rho_t\Big],
\qquad
\rho_t:=\max_a\big\{\mu_a(x_t)-\eta Q_{t,a}\big\}-\big\{\mu_{a_t}(x_t)-\eta Q_{t,a_t}\big\}\ \ge\ 0 .
\end{align*}
$R_T$ is the reward regret against the constrained optimum. $\tilde R_T$ is the regret of
the selection rule with respect to the penalized reward $\mu_a(x)-\eta Q_{t,a}$, the
quantity that standard contextual bandit theory controls. Lemma~\ref{lem:queue} and
Theorem~\ref{thm:dpp} hold for fixed reward functions $\{\mu_a\}$; in the Bayesian setting
of Proposition~\ref{prop:inherit} they hold conditionally on the true parameter, and hence
also in expectation over the prior.

\begin{lemma}[Deterministic queue bound]
\label{lem:queue}
For any $\eta>0$ and any estimates $\mu_{a,t}(\cdot)\in[0,1]$,
\[
Q_{t,a}\ \le\ C_\eta:=A(1+\eta^{-1})+2
\qquad\text{for all } t\ge1 \text{ and all } a .
\]
\end{lemma}

\begin{proof}
Let $m_t:=\min_a Q_{t,a}$ and $M_t:=\max_{s\le t}m_s$, with $M_0:=0$.

\smallskip\noindent
\emph{Step 1 (a selected queue is close to the minimum).}
If $a_t=a$, then $\mu_{a,t}(x_t)-\eta Q_{t,a}\ge\mu_{b,t}(x_t)-\eta Q_{t,b}$ for
$b\in\arg\min_{b}Q_{t,b}$, hence $Q_{t,a}\le m_t+\eta^{-1}$. Unselected queues do not
increase. By induction from $Q_{1,a}=0$,
\begin{equation}
\label{eq:maxq}
\max_a Q_{t,a}\ \le\ M_{t-1}+\eta^{-1}+1 \qquad\text{for all } t .
\end{equation}

\smallskip\noindent
\emph{Step 2 (the minimum is controlled by the maximum).}
Since $\sum_a\alpha_a=1$,
\[
\sum_a Q_{t+1,a}-\sum_a Q_{t,a}\;=\;\sum_{a\neq a_t}\big[\alpha_a-Q_{t,a}\big]_+ ,
\]
so the total queue length is non-decreasing and increases only at rounds with $m_t<1$.
Fix $t$ with $\sum_a Q_{t,a}>0$ and let $s<t$ be the last round at which the total
increased, so $\sum_a Q_{t,a}=\sum_a Q_{s+1,a}$. An agent $b$ with $Q_{s,b}=m_s<1$
has $Q_{s+1,b}<2$, and every other queue is at most $M_s+\eta^{-1}+1$ by \eqref{eq:maxq}.
Hence
\begin{equation}
\label{eq:minq}
A\,m_t\ \le\ \sum_a Q_{t,a}\ \le\ (A-1)\big(M_{t-1}+\eta^{-1}+1\big)+2 .
\end{equation}

\smallskip\noindent
\emph{Step 3 (fixed point).}
Let $\bar C_\eta:=(A-1)(1+\eta^{-1})+2$. If $M_{t-1}\le \bar C_\eta$, then \eqref{eq:minq}
gives $m_t\le \bar C_\eta$, so $M_t\le \bar C_\eta$. Since $M_0=0$, $M_t\le \bar C_\eta$ for
all $t$, and \eqref{eq:maxq} yields $Q_{t,a}\le \bar C_\eta+\eta^{-1}+1=C_\eta$.
\end{proof}

\begin{theorem}[Finite-time guarantees]
\label{thm:dpp}
For any $\eta>0$ and $T\ge1$:
\begin{enumerate}
\item[(a)] \emph{(Capacity)} For every agent $a$, deterministically,
\[
\frac1T\sum_{t=1}^T I_{t,a}-\alpha_a\ \le\ \frac{C_\eta}{T}.
\]
Since these deviations sum to zero over $a$, also
$\big|\frac1T\sum_{t}I_{t,a}-\alpha_a\big|\le (A-1)C_\eta/T$.
\item[(b)] \emph{(Reward)} $R_T\le \eta T+\tilde R_T$.
\end{enumerate}
\end{theorem}

\begin{proof}
\emph{(a)} Since $[z]_+\ge z$ and $Q_{1,a}=0$, $Q_{T+1,a}\ge\sum_{t\le T}(I_{t,a}-\alpha_a)$;
apply Lemma~\ref{lem:queue}.

\smallskip\noindent
\emph{(b)} Let $L_t:=\frac12\sum_a Q_{t,a}^2$ and $V:=\eta^{-1}$.

\smallskip\noindent
\emph{Drift.}
Using $[z]_+^2\le z^2$ and
$\frac12\sum_a(I_{t,a}-\alpha_a)^2=\frac12\big(1-2\alpha_{a_t}+\sum_a\alpha_a^2\big)\le1$,
\begin{equation}
\label{eq:drift}
L_{t+1}-L_t\ \le\ 1+\sum_a Q_{t,a}(I_{t,a}-\alpha_a)\ =\ 1+Q_{t,a_t}-\sum_a\alpha_a Q_{t,a}.
\end{equation}

\smallskip\noindent
\emph{Comparison with the constrained optimum.}
Let $a_t^\star\sim P_{\phi^\star}(\cdot\mid x_t)$ be the action of the optimal constrained
policy, drawn independently of the algorithm. By definition of $\rho_t$,
\[
\eta Q_{t,a_t}-\mu_{a_t}(x_t)\ \le\ \eta Q_{t,a_t^\star}-\mu_{a_t^\star}(x_t)+\rho_t .
\]
Multiplying by $V$ and combining with \eqref{eq:drift},
\begin{equation}
\label{eq:dpp}
L_{t+1}-L_t-V\mu_{a_t}(x_t)\ \le\ 1+Q_{t,a_t^\star}-\sum_a\alpha_a Q_{t,a}
-V\mu_{a_t^\star}(x_t)+V\rho_t .
\end{equation}

\smallskip\noindent
\emph{Conditioning.}
$Q_t$ is $\mathcal F_{t-1}$-measurable and $x_t$ is independent of $\mathcal F_{t-1}$, so
\[
\mathbb E\big[Q_{t,a_t^\star}\,\big|\,\mathcal F_{t-1}\big]
=\sum_a Q_{t,a}\,\mathbb E_x\big[P_{\phi^\star}(a\mid x)\big]
=\sum_a\alpha_a Q_{t,a},
\qquad
\mathbb E\big[\mu_{a_t^\star}(x_t)\,\big|\,\mathcal F_{t-1}\big]=\mu^\star .
\]
Taking expectations in \eqref{eq:dpp}, summing over $t\le T$, and using
$L_1=0\le L_{T+1}$ gives $V R_T\le T+V\tilde R_T$, i.e.\ $R_T\le\eta T+\tilde R_T$.
\end{proof}

\begin{proposition}[Bounding $\tilde R_T$ for Thompson sampling]
\label{prop:inherit}
Let $\theta^\star:=(\theta^\star_1,\dots,\theta^\star_A)$ be the true agent parameters of
Section~\ref{sec:learning_policy}, drawn from a known prior, and let
$\theta_t:=(\tilde\theta_{1,t},\dots,\tilde\theta_{A,t})$ be the posterior sample used at
round $t$. Suppose $\mu_a(x)=\mu_a(x;\theta^\star)$, $r_{t,a}\in\{0,1\}$, and
$\mu_{a,t}(\cdot)=\mu_a(\cdot;\theta_t)$ with $\theta_t$ drawn from the exact posterior given
$\mathcal F_{t-1}$. Then $\tilde R_T$ satisfies the Bayesian regret bounds of
\citet{neu2022lifting}: $O(\sqrt{AT\,H(\theta^\star)})$ for a finite-entropy prior, and
$\tilde O(A\sqrt{dT})$ for the logistic model of Section~\ref{sec:learning_policy} with
Lipschitz logits, where the joint parameter $\theta^\star\in\mathbb R^{Ad}$.
\end{proposition}

\begin{proof}
Let $\tilde x_t:=(x_t,Q_t)$ and
$\bar\mu_a(\tilde x_t;\theta):=\mu_a(x_t;\theta)-\eta Q_{t,a}$, and let
$a_t^\dagger:=\arg\max_a\bar\mu_a(\tilde x_t;\theta^\star)$, so that
$\rho_t=\bar\mu_{a_t^\dagger}(\tilde x_t;\theta^\star)-\bar\mu_{a_t}(\tilde x_t;\theta^\star)$.
Since $Q_t$ is $\mathcal F_{t-1}$-measurable, $\tilde x_t$ is a context selected by an
adaptive adversary, a setting covered by \citet{neu2022lifting}. The penalty is known
and independent of $\theta^\star$, so the posterior over $\theta^\star$ coincides with that
of the unpenalized problem. Hence $a_t=\arg\max_a\bar\mu_a(\tilde x_t;\theta_t)$ is exact
Thompson sampling for the reward $\bar\mu$, and $\tilde R_T$ is its Bayesian regret.
Because $\theta_t$ and $\theta^\star$ have the same distribution given $\mathcal F_{t-1}$
and $\tilde x_t$, the pairs $(a_t,\theta_t)$ and $(a_t^\dagger,\theta^\star)$ are
identically distributed, giving the decoupling identity
\[
\mathbb E\big[\bar\mu_{a_t^\dagger}(\tilde x_t;\theta^\star)-\bar\mu_{a_t}(\tilde x_t;\theta^\star)\big]
=\mathbb E\big[\bar\mu_{a_t}(\tilde x_t;\theta_t)-\bar\mu_{a_t}(\tilde x_t;\theta^\star)\big]
=\mathbb E\big[\mu_{a_t}(x_t;\theta_t)-\mu_{a_t}(x_t;\theta^\star)\big],
\]
in which the per-action shift $\eta Q_{t,a}$ cancels, so the information-ratio bound for
rewards in $[0,1]$ applies unchanged.
\end{proof}

\begin{corollary}
Under the assumptions of Proposition~\ref{prop:inherit} with $\eta=\sqrt{A/T}$,
\[
R_T= O\big(\sqrt{AT}\,(1+\sqrt{H(\theta^\star)})\big),
\]
and every agent's realized assignment fraction is within $O(\sqrt{A/T})$ of $\alpha_a$
(for $T\ge A$).
\end{corollary}

\paragraph{Remarks.}
\begin{enumerate}
\item[(i)] Lemma~\ref{lem:queue} and Theorem~\ref{thm:dpp}(a) require only
$\mu_{a,t}\in[0,1]$, so they hold for all four deployed policies (including greedy and the
Laplace/bootstrap approximations) at any fixed $\eta$. At $\eta=0.5$ and $A=2$,
$C_\eta=8$, giving a deviation of at most $8/T$: $0.004$ for $T\ge 2000$ and $0.037$ for
MMLU ($T=217$). The observed maximum deviations are $0.0011$ and $0.0113$
(Table~\ref{tab:capacity_all_datasets}). The lemma assumes every agent carries a queue
with $\sum_a\alpha_a=1$ and therefore does not cover the unconstrained-agent extension of
Appendix~\ref{app:with_free_agent}.
\item[(ii)] Theorem~\ref{thm:dpp}(b) is sublinear only if $\eta\to0$ with $T$. At fixed
$\eta$, the $\eta T$ term reflects the granularity of the queue-implied shadow prices
$\eta Q_{t,a}$, which move in increments of order $\eta$. Appendix~\ref{app:diff_eta}
reports the corresponding reward--capacity trade-off empirically; the deployed
$\eta=0.5$ accepts an $O(\eta)$ per-round reward term in exchange for tighter capacity
tracking.
\item[(iii)] Proposition~\ref{prop:inherit} does not cover greedy selection or approximate
posteriors. For those, Theorem~\ref{thm:dpp}(b) still holds with $\tilde R_T$ left
as the estimation-error term.
\end{enumerate}

\section{Experiment Details and Additional Results}
\label{app:exp}
All experiments are run on 2 x AMD Epyc 7313 (16 cores) CPUs. No GPU or major RAM was required.

\subsection{Dataset Details}
\label{app:dataset}

Here we provide additional details on the datasets used in Section~\ref{sec:exp} and Appendix~\ref{app:additional_datasets}. The tabular datasets are publicly available from the following sources: 
\href{https://archive.ics.uci.edu/dataset/222/bank+marketing}{Bank Marketing} (i.e., \textbf{Bank}),
\href{https://archive.ics.uci.edu/dataset/350/default+of+credit+card+clients}{Default of Credit Card Clients} (i.e., \textbf{Credit}),
\href{https://archive.ics.uci.edu/dataset/603/in+vehicle+coupon+recommendation}{Vehicle Coupon Recommendation} (i.e., \textbf{Coupon}), and
\href{https://www.kaggle.com/datasets/sulianova/cardiovascular-disease-dataset}{Cardiovascular Disease} (i.e., \textbf{Cardio}), \href{https://huggingface.co/datasets/cais/mmlu}{\textbf{MMLU}},\href{https://osf.io/2ntrf/overview?view_only=9ec9cacb806d4a1ea4e2f8acaada8f6c}{\textbf{ImageNet16H}} , \href{https://huggingface.co/datasets/shlomihod/civil-comments-wilds}{CivilComments-WILDS}(i.e., \textbf{Civil})), and \href{https://camelyon17.grand-challenge.org/Data/}{\textbf{Camelyon17}}. 

Table~\ref{apptab:datasets} summarizes the datasets and agents' marginal performance.  
\begin{table}[htbp]
\centering
\caption{Summary of datasets and agent performance. We use a subsample of each dataset for online evaluation under a train–test split; the resulting number of timesteps/tasks is denoted by $T$ and varies across datasets.}
\label{apptab:datasets}
\begin{tabular}{lcccc}
\toprule
Dataset & \# Features & \# Timesteps/Tasks & Agent 1 Acc. & Agent 2 Acc. \\
\midrule
\textbf{Bank}   & 15 & 12,357 & 0.91 & 0.53 \\
\textbf{Credit} & 23 & 9,000  & 0.82 & 0.67 \\
\textbf{Coupon} & 24 & 3,624  & 0.72 & 0.57 \\
\textbf{Cardio} & 11 & 21,000 & 0.74 & 0.58 \\
\textbf{MMLU}    & 1 & 217    & 0.32   & 0.47   \\
\textbf{ImageNet16H} & 1 & 2000 (200 images × duplicated 10x)
 & 0.77 & 0.86 \\
\textbf{Camelyon17} & 1 & 2000 & 0.69 & 0.69 \\
\textbf{Civil} & 10 & 5880 & 0.58 & 0.56\\
\bottomrule
\end{tabular}
\end{table}

For the \textbf{MMLU} experiments, we use multiple choice questions from 2 categories: \textit{college chemistry} and \textit{US foreign policy}. We use the \href{https://huggingface.co/sentence-transformers/all-MiniLM-L6-v2}{sentence-transformers/all-MiniLM-L6-v2} model from HuggingFace to transform the text to a 384-dimensional embedding. We then apply principal component analysis to the high-dimensional embedding to identify the first principal component. This is used as the single context feature for the learning policies. The reward is binary since the agent is either correct or incorrect.

For CivilComments-WILDS (i.e., \textbf{Civil}), we study toxicity classification under subpopulation distribution shift. Each comment is annotated with the demographic identities it mentions, and toxicity classifiers can exhibit different error patterns across these identity subpopulations. We use pretrained LLM inference results from the HELM database \cite{liang2022holistic} and construct the task context by applying PCA to sentence embeddings, retaining the first 10 principal components. We consider two model pairs exhibiting heterogeneous expertise across demographic subpopulations. We use \texttt{together-redpajama-incite-instruct-3b-v1} and
\texttt{meta-llama-7b}. The relative accuracy of the two agents reverses across the corresponding subpopulations. The reward is binary, indicating whether the assigned model correctly predicts the toxicity label.

For \textbf{ImageNet16H}, we use a VGG-19 Convolutional Neural Network fine-tuned for 10 epochs as one arm, and the predictions of Human 48 as the other. This is a multi-class classification problem with 16 classes. We assume symmetric misclassification utilities, which allows us to use a binary reward signal for our policies. We use \textit{phase noise}, a feature in the dataset, as our singular context feature. 

We train two linear classifiers as the two arms for \textbf{Camelyon17}, using data from two different hospitals to induce distribution shift and heterogeneous expertise. The binary classification task is detecting the presence/absence of a tumor. We extract 21 features from the pixels in the raw image. We identify the first principal component and use that as context in our model. The features used for the \textbf{Camelyon17} dataset are as follows:

\begin{enumerate}
    \item Mean pixel intensity (R channel)
    \item Standard deviation of pixel intensity (R channel)
    \item 10th percentile of pixel intensity (R channel)
    \item 90th percentile of pixel intensity (R channel)
    \item Maximum pixel intensity (R channel)
    \item Mean pixel intensity (G channel)
    \item Standard deviation of pixel intensity (G channel)
    \item 10th percentile of pixel intensity (G channel)
    \item 90th percentile of pixel intensity (G channel)
    \item Maximum pixel intensity (G channel)
    \item Mean pixel intensity (B channel)
    \item Standard deviation of pixel intensity (B channel)
    \item 10th percentile of pixel intensity (B channel)
    \item 90th percentile of pixel intensity (B channel)
    \item Maximum pixel intensity (B channel)
    \item Whiteness fraction (proportion of pixels with R, G, B $> 220$; proxy for background)
    \item A proxy for Hematoxylin: mean of $B - R$ (purpleness; indicative of nuclear staining)
    \item A proxy for Eosin: mean of $R - B$ (pinkness; indicative of cytoplasm staining)
    \item Mean intensity across all channels
    \item Colour contrast: standard deviation of grayscale intensity
    \item Mean G channel intensity
\end{enumerate}

For four tabular datasets (\textbf{Bank}, \textbf{Credit}, \textbf{Coupon}, \textbf{Cardio}), we construct two classifier agents as two arms. Agent 1 is a logistic regression model trained on a randomly selected subset of up to 20\% of the available features, while Agent 2 is an XGBoost classifier trained on the remaining features. This asymmetric feature access and model class induce differences in marginal accuracy and context-dependent expertise. The contextual assignment policy observes the full feature space and uses observed rewards to learn which agent is better suited for each context. For the logistic regression agent, we use the \texttt{lbfgs} solver with balanced class weights and a maximum of 1000 iterations. For the XGBoost agent, we use 200 trees with maximum depth 4, learning rate 0.05, subsampling rate 0.8, column subsampling rate 0.8, and log-loss evaluation, with the random seed fixed to the experiment seed.

\subsection{Additional Dataset Results}
\label{app:additional_datasets}

To conserve space in the main text, we report the full capacity-sweep
results here for the remaining four datasets introduced in
Section~\ref{subsec:main_results}: three tabular datasets, \textbf{Credit}
\citep{yeh2009comparisons} (23 features, 9,000 tasks, predicting default),
\textbf{Coupon} \citep{wang2017bayesian} (24 features, 3,624 tasks,
predicting coupon acceptance), and \textbf{Cardio} (11 features, 21,000
tasks, predicting cardiovascular disease); and one text-based dataset,
\textbf{Civil}, a toxicity classification benchmark with demographic
subpopulation shift. These follow the same experimental protocol as Section~\ref{subsec:main_results}, see
Appendix~\ref{app:dataset} for details.

Figure~\ref{fig:additional_datasets} reports the average error rate as a
function of Agent 1's capacity $\alpha_1$ for all four datasets. As in the
main text, contextual policies (logistic and tree-based, both Thompson
Sampling and Greedy) consistently outperform the non-contextual baseline across all capacity levels, tree-based models generally outperform logistic
models, and Thompson Sampling and Greedy policies perform similarly, consistent with the implicit exploration induced by the capacity constraints discussed in Section~\ref{subsec:main_results}.

\begin{figure}[t]
    \centering

    \begin{subfigure}[t]{0.45\linewidth}
        \centering
        \includegraphics[width=\linewidth]{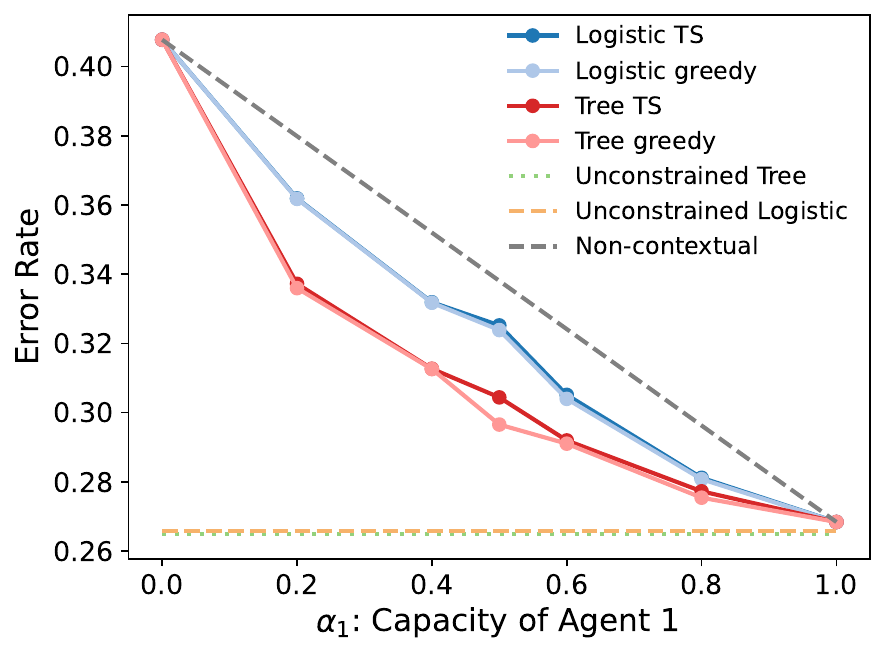}
        \caption{Cardio}

    \end{subfigure}
    \hfill
    \begin{subfigure}[t]{0.45\linewidth}
        \centering
        \includegraphics[width=\linewidth]{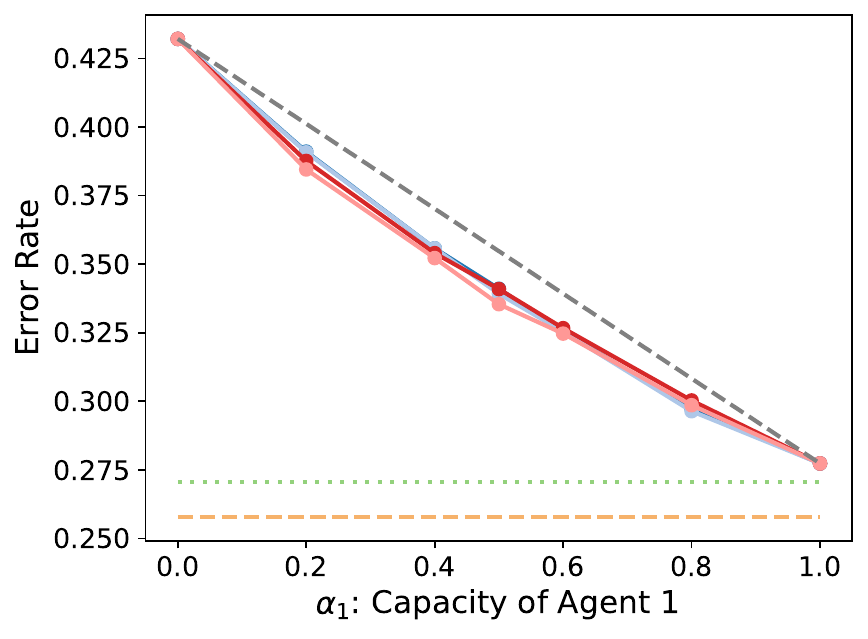}
        \caption{Coupon}

    \end{subfigure}

    \vspace{0.5em}

    \begin{subfigure}[t]{0.45\linewidth}
        \centering
        \includegraphics[width=\linewidth]{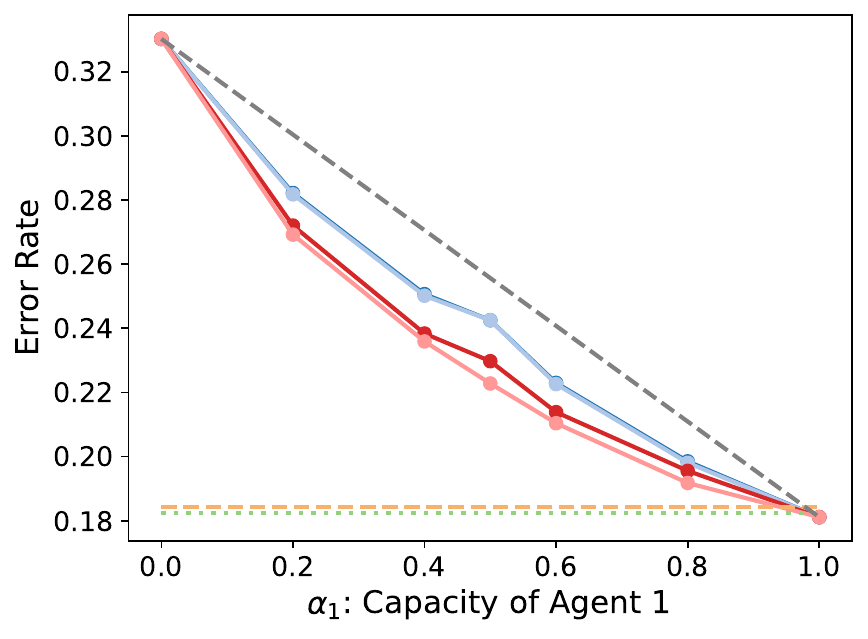}
        \caption{Credit}

    \end{subfigure}
    \hfill
    \begin{subfigure}[t]{0.45\linewidth}
        \centering
        \includegraphics[width=\linewidth]{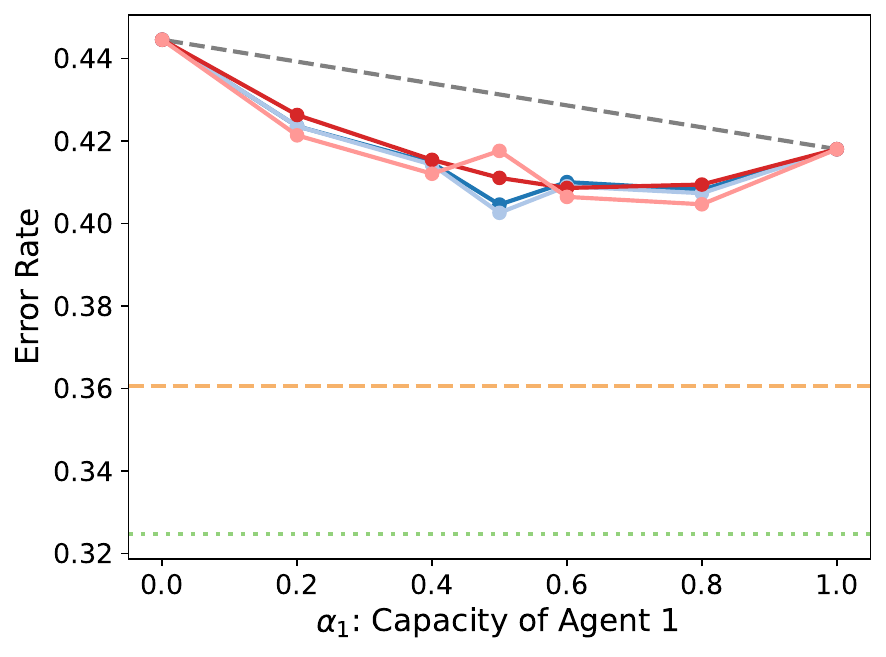}
        \caption{Civil}

    \end{subfigure}

    \caption{Error rate as a function of capacity (two-agent case), for the four
    additional datasets not shown in Figure~\ref{fig:main_results}. Points
    represent the average error rate over 100 runs. Uncertainty estimates
    are reported in Appendix~\ref{app:uncertainty}.}
    \label{fig:additional_datasets}
\end{figure}
\color{black}

\subsection{Uncertainty Across Randomized Task Streams}
\label{app:uncertainty}

The main results in Figures~\ref{fig:camerelyon},  \ref{fig:main_results}, \ref{fig:additional_datasets}
are averaged over 100 randomized permutations of the online task sequence. Table~\ref{tab:uncertainty_summary} summarizes uncertainty across these runs. For each dataset, we report the range of error reductions relative to the non-contextual baseline across the four contextual policies and five evaluated capacity levels, together with the largest standard error of the mean error rate across these settings.

\begin{table}[t]
    \centering
    \caption{Uncertainty across randomized task streams. Error reductions and
    standard errors are reported in percentage points. Each row summarizes four
    contextual policies and five evaluated capacity levels over 100 randomized
    task streams.}
    \label{tab:uncertainty_summary}
    \small
    \begin{tabular}{lcc}
        \toprule
        Dataset & Error reduction vs.\ non-contextual & Max.\ SE \\
        \midrule
        Bank          & 4.80--10.70 & 0.031 \\
        Credit        & 1.30--4.40  & 0.070 \\
        Coupon        & 0.60--1.80  & 0.064 \\
        Cardio        & 1.30--4.40  & 0.070 \\
        MMLU          & 2.47--6.15  & 0.290 \\
        ImageNet16H   & 1.26--4.21  & 0.090 \\
        Camelyon17    & 4.25--8.03  & 0.080 \\
        Civil    & 4.07–6.65  & 0.001 \\
        \bottomrule
    \end{tabular}
\end{table}

Across all datasets, the standard errors are small relative to the
observed improvements over non-contextual assignment, indicating that the
reported gains are stable across randomized task orderings.

\subsection{Additional Results: Tabular Datasets Across Random Seeds}
\label{app:tab_more_seeds}

We provide additional robustness checks using multiple random seeds. We use the same agent-construction procedure as in Section~\ref{subsec:main_results}: for each dataset, one agent is a logistic regression model trained on a random subset of up to 20\% of the features, while the other is an XGBoost classifier trained on the remaining features.

While Figure~\ref{fig:main_results} reports results for one seed, Table~\ref{apptab:multi_seed_results} reports results across additional seeds for \textbf{Bank}, \textbf{Credit}, \textbf{Coupon}, and \textbf{Cardio}. Across seeds, online contextual policies consistently outperform the non-contextual baseline. Tree-based greedy policies generally achieve comparable or lower error than logistic greedy policies. The magnitude of the improvement varies across seeds because the induced agent expertise depends on the random feature split; when the two agents have more similar expertise, the gains from contextual assignment are smaller.

\begin{table}[htbp]
\centering
\caption{Error rates across datasets and seeds, average over 10 runs, for fixed capacity with two agents $\alpha1=\alpha2=0.5$}
\label{apptab:multi_seed_results}
\begin{tabular}{llccc}
\toprule
Dataset & Seed & Non-contextual & Logistic Greedy & Tree Greedy \\
\midrule
\multirow{5}{*}{Bank}
& 0 & 0.128 & 0.122 & 0.113 \\
& 1 & 0.136 & 0.129 & 0.117 \\
& 2 & 0.278 & 0.195 & 0.171 \\
& 3 & 0.260 & 0.187 & 0.180 \\
& 4 & 0.124 & 0.097 & 0.097 \\
\midrule
\multirow{5}{*}{Credit}
& 0 & 0.195 & 0.190 & 0.188 \\
& 1 & 0.194 & 0.190 & 0.188 \\
& 2 & 0.256 & 0.241 & 0.222 \\
& 3 & 0.257 & 0.242 & 0.224 \\
& 4 & 0.197 & 0.192 & 0.188 \\
\midrule

\multirow{5}{*}{Coupon}
& 0 & 0.352 & 0.343 & 0.342 \\
& 1 & 0.369 & 0.357 & 0.360 \\
& 2 & 0.357 & 0.335 & 0.341 \\
& 3 & 0.347 & 0.331 & 0.340 \\
& 4 & 0.349 & 0.341 & 0.345 \\
\midrule

\multirow{5}{*}{Cardio}
& 0 & 0.334 & 0.332 & 0.309 \\
& 1 & 0.373 & 0.354 & 0.309 \\
& 2 & 0.338 & 0.335 & 0.306 \\
& 3 & 0.340 & 0.338 & 0.307 \\
& 4 & 0.343 & 0.338 & 0.311 \\

\bottomrule
\end{tabular}
\end{table}

\subsection{Additional Results: Multi-Agent Experiments with Capacity Constraints}
\label{app:multi-agent}
We extend the two-agent setting to multi-agent assignment and evaluate the proposed policies on the \textbf{MMLU} dataset. Specifically, we consider settings with $A=3$ agents in Figure~\ref{fig:3agent} and $A=5$ agents in Figure~\ref{fig:5agent}. In each experiment, we vary the capacity assigned to one agent, LLaMa-2-70B, and allocate the remaining capacity equally across the other agents.

Across both multi-agent settings, the contextual policies consistently achieve lower average error rates than the non-contextual assignment baseline. These results show that the benefits of contextual assignment are not limited to the two-agent case: when capacity must be allocated across multiple agents, learning context-dependent expertise still enables the policy to better match tasks to agents while respecting capacity constraints. This provides additional empirical support for the theoretical analysis in Section~\ref{subsec:cap_constraint} and the main results in Section~\ref{subsec:main_results}.

\begin{figure}[htbp]
    \centering
    \begin{subfigure}[b]{0.49\textwidth}
        \centering
        \includegraphics[width=\textwidth]{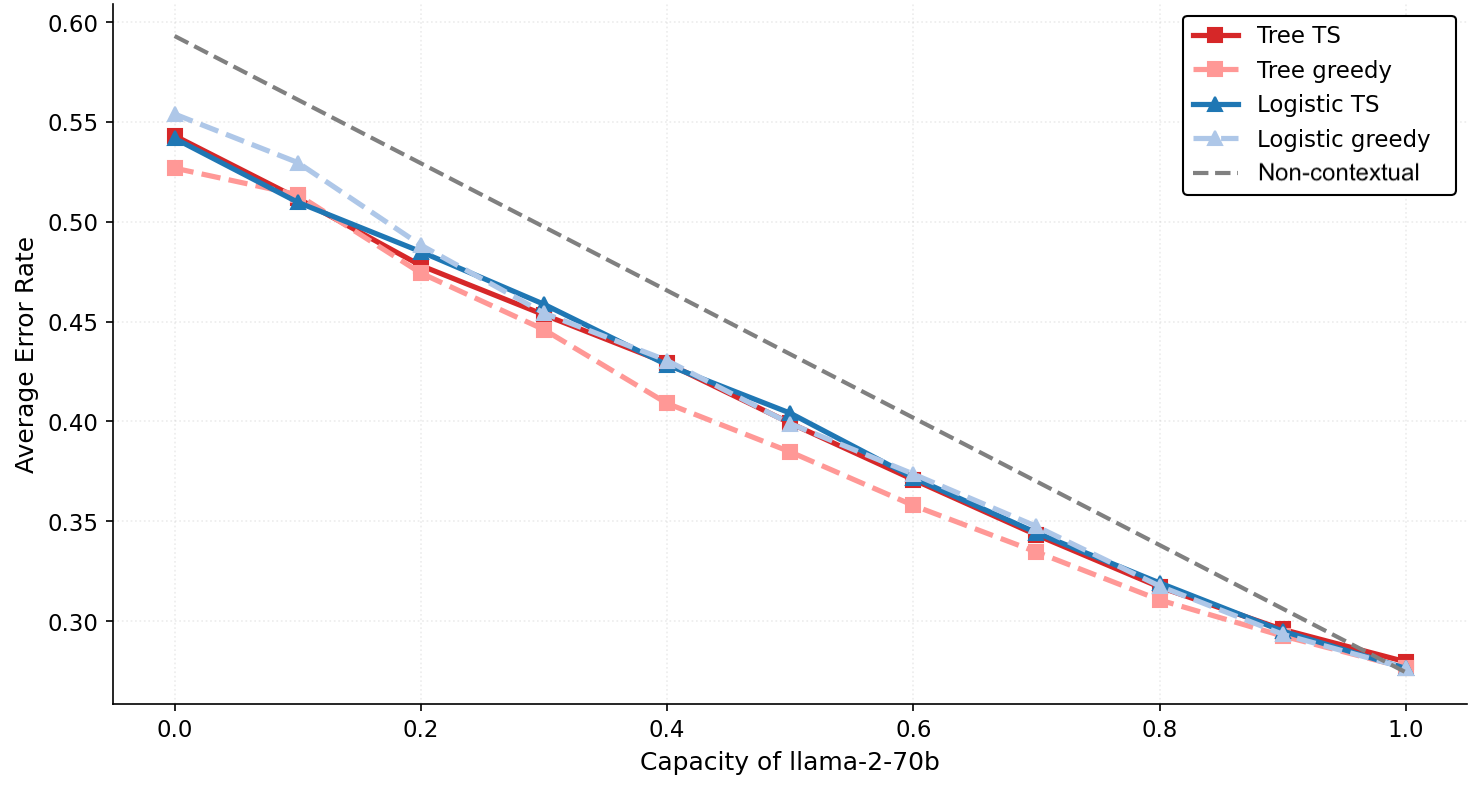}
        \caption{}
        \label{fig:3agent}
    \end{subfigure}
    \hfill 
    \begin{subfigure}[b]{0.49\textwidth}
        \centering
        \includegraphics[width=\textwidth]{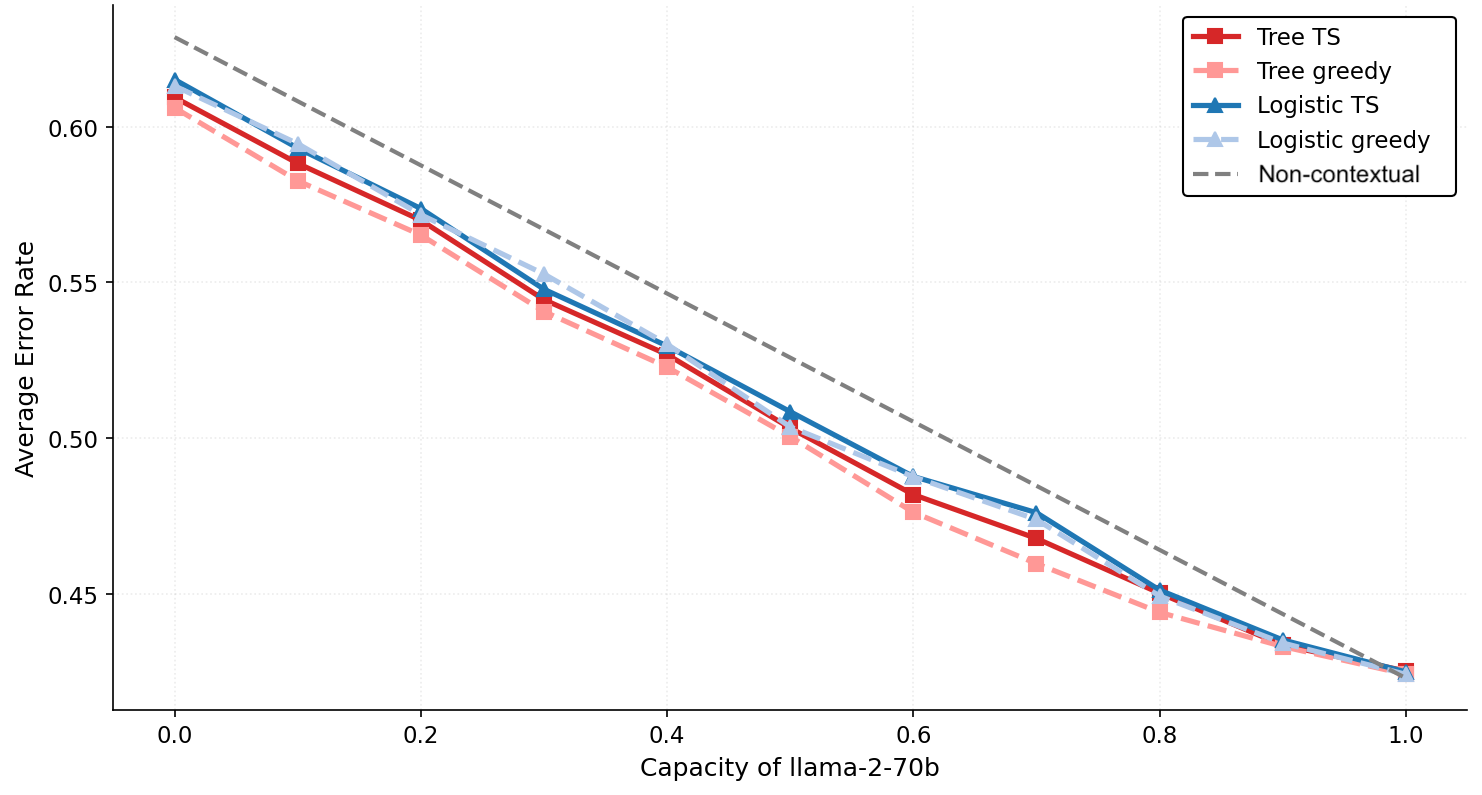}
        \caption{}
        \label{fig:5agent}
    \end{subfigure}
    \caption{Multi-agent assignment results on \textbf{MMLU}. (a) evaluates three LLM agents (Cohere Small (20220720), LLaMa-13B, and LLaMa-2-70B) on College Chemistry, Computer Security, and US Foreign Policy. (b) evaluates five LLM agents (Cohere Small (20220720), LLaMa-13B, LLaMa-2-70B, command-medium, and t0pp) on Abstract Algebra, College Chemistry, Computer Security, Econometrics, and US Foreign Policy. Results are averaged over 25 runs.}

    \label{fig:combined_agents}
\end{figure}

\subsection{Additional Results: Extensions with an Additional Unconstrained Agent}
\label{app:with_free_agent}
In some settings, only a subset of agents face meaningful capacity constraints. For example, in a hospital setting, two doctors may have limited bandwidth, while a machine learning model can process cases at negligible marginal cost. To capture this setting, we consider an extension with two constrained agents and one unconstrained ``free'' agent. Let agents $a \in \{1,2\}$ be capacity-constrained, with target capacities $\alpha_1$ and $\alpha_2$, where $\alpha_1+\alpha_2=1$, and let agent $a=3$ denote the unconstrained agent. The online assignment rule remains as \ref{eq:agent_select_w_cap} that
\[
a_t = \arg\max_{a \in \{1,2,3\}} \left\{ \mu_{a,t}(x_t) - \eta Q_{t,a} \right\},
\]
where $\mu_{a,t}(x_t)$ is the current estimate of agent $a$'s expected reward for context $x_t$. For the constrained agents, the queues evolve as
\[
Q_{t+1,a} = \bigl[Q_{t,a} + \mathbbm{1}\{a_t=a\} - \alpha_a\bigr]_+,
\qquad a \in \{1,2\}.
\]
The unconstrained agent has no capacity penalty, which is equivalent to setting $Q_{t,3}=0 \quad \text{for all } t$.
Thus, the free agent enters the same assignment rule as the constrained agents, but without a queue penalty (i.e., $a_t
=
\arg\max_{a \in \{1,2,3\}}
\left\{
\mu_{a,t}(x_t) - \eta Q_{t,a}
\right\}, Q_{t,3}=0$).

In our implementation, the constrained agents are updated exactly as in the main setting, and the free agent's reward model is updated whenever it is selected. The capacity queues are updated only for the constrained agents. This design allows the policy to route tasks either to capacity-limited agents or to the unconstrained model, while ensuring that the constrained agents do not exceed their target capacities.

We evaluate the extension using the \textbf{Bank} dataset and the same two constrained agents as in Section~\ref{subsec:main_results}. We add one unconstrained ``free agent'' with a marginal error rate $0.30$. Results are averaged over 50 runs with $\eta=0.5$. Figure~\ref{fig:plus_one_agent} reports the average error rate as the capacity of Agent 1 varies from 0 to 1. All contextual policies outperform both the always querying free-agent baseline and the two-agent non-contextual random baseline. This shows that the unconstrained agent can improve performance by handling tasks for which the capacity-constrained agents have low adjusted scores (i.e., $\mu_{a,t}(x_t)-\eta Q_{t,a}$).

Unlike the two-agent setting, the endpoints do not correspond to the marginal error rate of a single constrained agent. When one constrained agent has zero capacity, the policy still chooses between the other constrained agent and the free agent. The policies also do not converge to the same endpoint performance, since different learning rules estimate agent expertise differently.
\begin{figure}[htbp]
    \centering
    \includegraphics[width=0.5\linewidth]{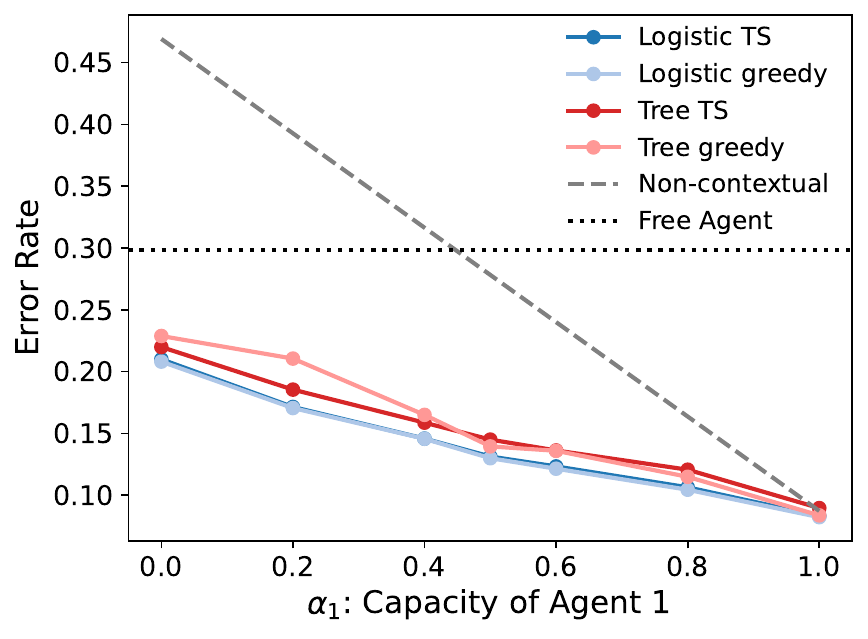}
    \caption{Extension with an unconstrained ``free agent'' and two capacity-constrained agents on the \textbf{Bank} dataset. The figure reports average error rates over 50 runs as the capacity of Agent 1 varies.}
\label{fig:plus_one_agent}
\end{figure}

\subsection{Additional Results: Mini-Batch Setting}
\label{app:mini-batch}

The mini-batch setting differs from the main online setting only in how assignments and updates are grouped. Instead of assigning one task at a time, tasks are processed in batches of size $B$. For each batch, we compute estimated rewards $\mu_{a,t}(x)$ for every agent-task pair. A subtle distinction from the fully online setting is that batch-level assignments must use integer counts. Given target capacities $\{\alpha_a\}$, we first determine the number of tasks assigned to each agent in the current batch. These batch-level counts are chosen to approximate the target capacities while accounting for accumulated queue imbalances from previous batches. For example, if $B=11$ and $\alpha_1=\alpha_2=0.5$, one agent must receive 6 tasks and the other 5 tasks in that batch; the virtual queue tracks this discrepancy so that future batches can compensate.

Conditional on these batch-level counts, we solve the assignment problem using a min-cost max-flow (MCMF) algorithm with adjusted scores $\mu_{a,t}(x)-\eta Q_{t,a}$. The MCMF step assigns tasks to agents to maximize predicted reward subject to the hard batch-level count constraints. Because $Q_{t,a}$ is constant across tasks within a batch for a fixed agent, the queue term does not affect the within-batch matching conditional on fixed counts. Its main role is to adjust future batch counts when exact proportional allocation is not possible within a batch.

After the batch is assigned, rewards are observed, the reward models are updated using the observed outcomes, and the virtual queues are updated based on the number of tasks assigned to each agent in the batch:\[Q_{t+B,a}=\left[Q_{t,a}+N_{t,a}-B\alpha_a\right]_+,\]where $N_{t,a}$ is the number of tasks assigned to agent $a$ in the batch. Thus, capacity is enforced as a hard constraint within each batch through integer assignment counts, while the queue maintains long-run capacity balance across batches.

In addition to the tree-based greedy results reported in Figure~\ref{fig:mini_batch}, we report the corresponding logistic greedy results in Figure~\ref{appfig:mini_batch_logistic}. 
The patterns are consistent: mini-batch assignment with $B=100$ improves over the fully online setting by enabling joint allocation over multiple tasks. Varying the batch size again reveals a U-shaped pattern: very small batches provide frequent updates but less stable capacity allocation, while very large batches delay feedback and slow learning. Intermediate batch sizes achieve the lowest error. These improvements, however, come at the cost of increased latency.

\begin{figure}[htbp]
    \centering
    \begin{subfigure}{0.43\linewidth}
        \centering
        \includegraphics[width=\linewidth]{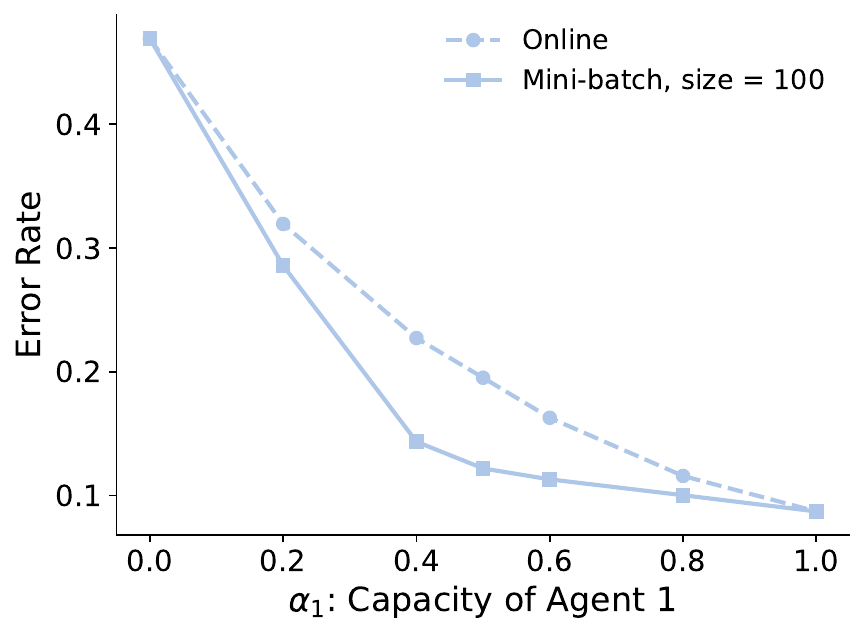}
        \caption{Online vs mini-batch assignment with batch size $N_B=100$. Error rates are averaged over 50 runs as the capacity of Agent 1 varies.}
    \end{subfigure}
    \hfill
    \begin{subfigure}{0.43\linewidth}
        \centering
        \includegraphics[width=\linewidth]{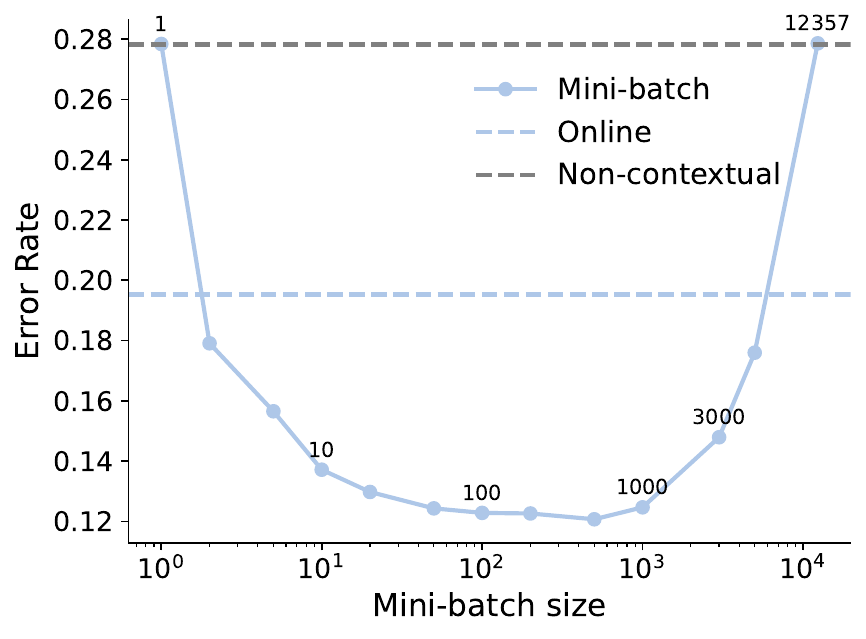}
        \caption{Effect of batch size when Agents 1 and 2 have equal capacities, $\alpha_1=\alpha_2=0.5$. Average error rates are computed over 10 runs. }
    \end{subfigure}

    \caption{Mini-batch assignment on the \textbf{Bank} dataset using the logistic greedy policy. }
    \label{appfig:mini_batch_logistic}
\end{figure}

\subsection{Additional  Results: Effect of the Queue Penalty Parameter \texorpdfstring{$\eta$}{eta}}
\label{app:diff_eta}
As discussed in Section~\ref{sec:learning_policy} and Appendix~\ref{app:capacity_constraints}, the queue penalty parameter $\eta$ controls the trade-off between reward maximization and capacity enforcement. The online policy selects
\[
a_t = \arg\max_{a \in \{1,\dots,A\}} \left\{ \mu_{a,t}(x_t) - \eta Q_{t,a} \right\},
\]
as in ~\eqref{eq:agent_select_w_cap}, where $\mu_{a,t}(x_t)$ is the current estimate of agent $a$'s expected reward for context $x_t$ and $Q_{t,a}$ is the virtual queue tracking capacity pressure. Smaller values of $\eta$ place more weight on estimated expertise, while larger values place more weight on satisfying the target capacities.

For all experiments in Section~\ref{sec:exp} and Appendix~\ref{app:additional_datasets}, we fix $\eta=0.5$. Here, we examine the sensitivity of performance to this choice. We use the \textbf{Camelyon17} setting from Section~\ref{subsec:main_results}, vary the capacity of Agent 1, and compare tree-based greedy policies with different values of $\eta$ against the non-contextual random baseline. Results are averaged over 100 runs. Figure~\ref{fig:diff_eta} shows the results. As expected, smaller values of $\eta$ lead to lower error rates because the assignment rule places more weight on estimated agent expertise and allows short-run deviations from the target capacities. When $\eta=0$, the policy converges to the unconstrained contextual policy, selecting agents purely based on estimated reward. As $\eta$ increases, the queue penalty plays a larger role in assignment decisions. For large values such as $\eta=5$, the capacity term dominates, and performance approaches the non-contextual random baseline because assignments are driven primarily by capacity balancing rather than contextual expertise.

Overall, $\eta$ determines how strongly the decision-maker prioritizes capacity enforcement relative to predictive performance. The appropriate choice depends on the application: settings with strict workload or fairness requirements may prefer larger $\eta$, while settings that tolerate short-run capacity deviations may benefit from smaller $\eta$.

\begin{figure}[htbp]
    \centering
    \includegraphics[width=0.60\linewidth]{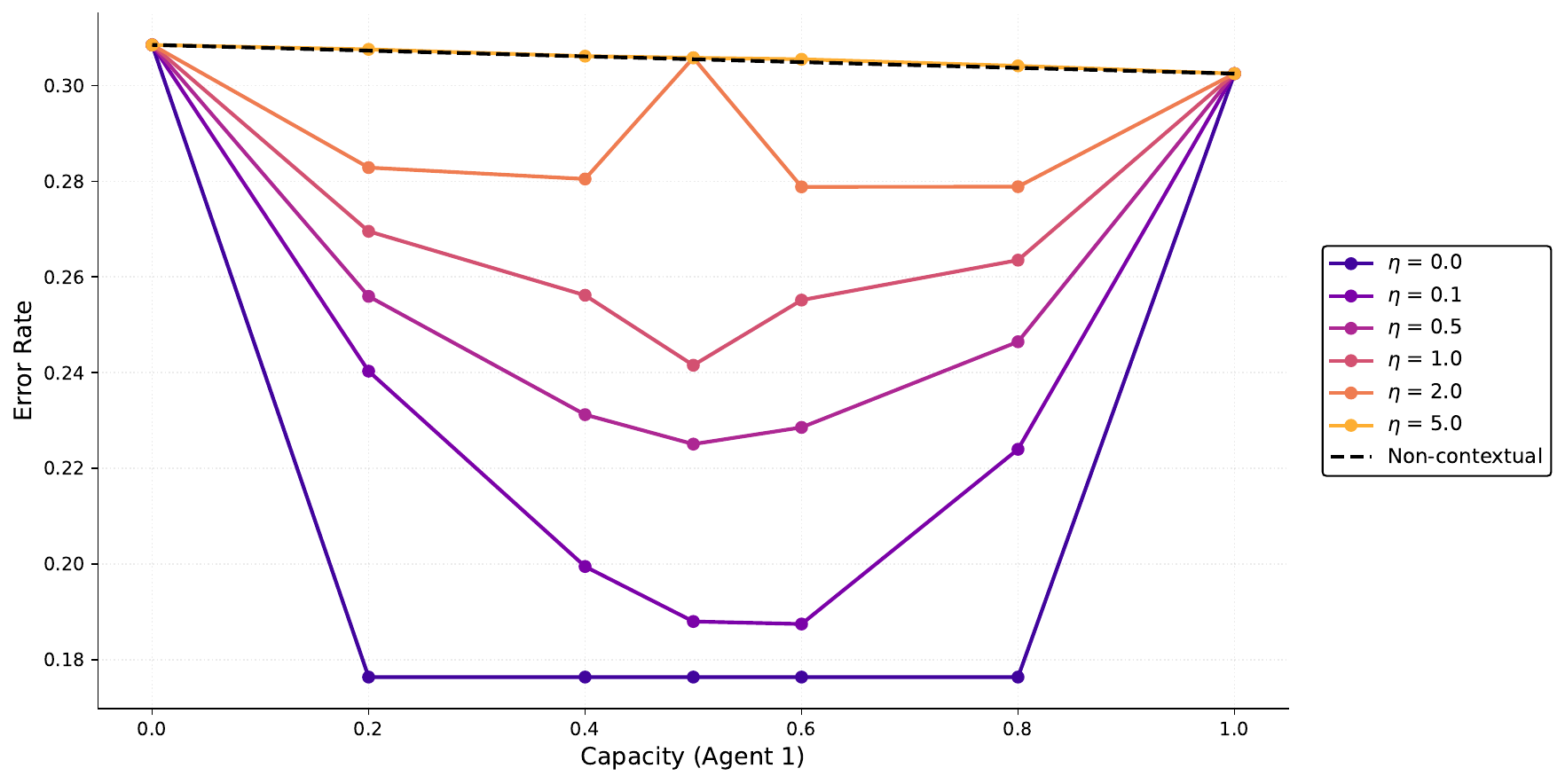}
    \caption{Effect of the queue penalty parameter $\eta$ in the \textbf{Camelyon} dataset (described in Section \ref{subsec:main_results}. Smaller $\eta$ prioritizes reward maximization and approaches the unconstrained contextual policy, while larger $\eta$ enforces capacity more aggressively and approaches the non-contextual random baseline.}
    \label{fig:diff_eta}
\end{figure}

\subsection{Validation of Capacity Constraint Enforcement}
\label{app:capacity_validation}

We validate the queue-based capacity mechanism in two complementary settings:
a controlled simulation where the optimal constrained allocation and shadow price
are analytically known, and a real-data experiment using the Bank dataset.

\paragraph{Synthetic validation.}
We first consider a setting with known agent expertise. We draw
$x\sim\mathcal{N}(0,1)$ and set
\[
    \mu_1(x)=\sigma(2x),
    \qquad
    \mu_2(x)=\sigma(-2x),
\]
with target capacities $(\alpha_1,\alpha_2)=(0.3,0.7)$ and $\eta=0.5$.
We run 50 independent task streams of $T=5{,}000$ tasks each. Because
$\mu_1(x)-\mu_2(x)$ is increasing in $x$, the optimal constrained policy
assigns Agent 1 to the largest $30\%$ of contexts. The corresponding threshold is
\[
    P(X\geq x^\star)=0.3
    \quad\Longrightarrow\quad
    x^\star=\Phi^{-1}(0.7)=0.524.
\]
At this threshold, the analytical shadow-price difference is
\[
    \lambda_1-\lambda_2
    =
    \mu_1(x^\star)-\mu_2(x^\star)
    =
    0.481.
\]

Figure~\ref{fig:capacity_simulation} shows that the online queue mechanism
closely reproduces both analytical predictions. The mean final realized
assignment fractions are $0.300$ and $0.700$, matching the target
capacities. Moreover, the queue-implied shadow-price difference,
$\eta(Q_{t,1}-Q_{t,2})$, averages $0.471$ over the final 500 tasks, close to
the analytical value $0.481$.

\begin{figure}[t]
    \centering
    \begin{subfigure}[t]{0.49\textwidth}
        \centering
        \includegraphics[width=\linewidth]
        {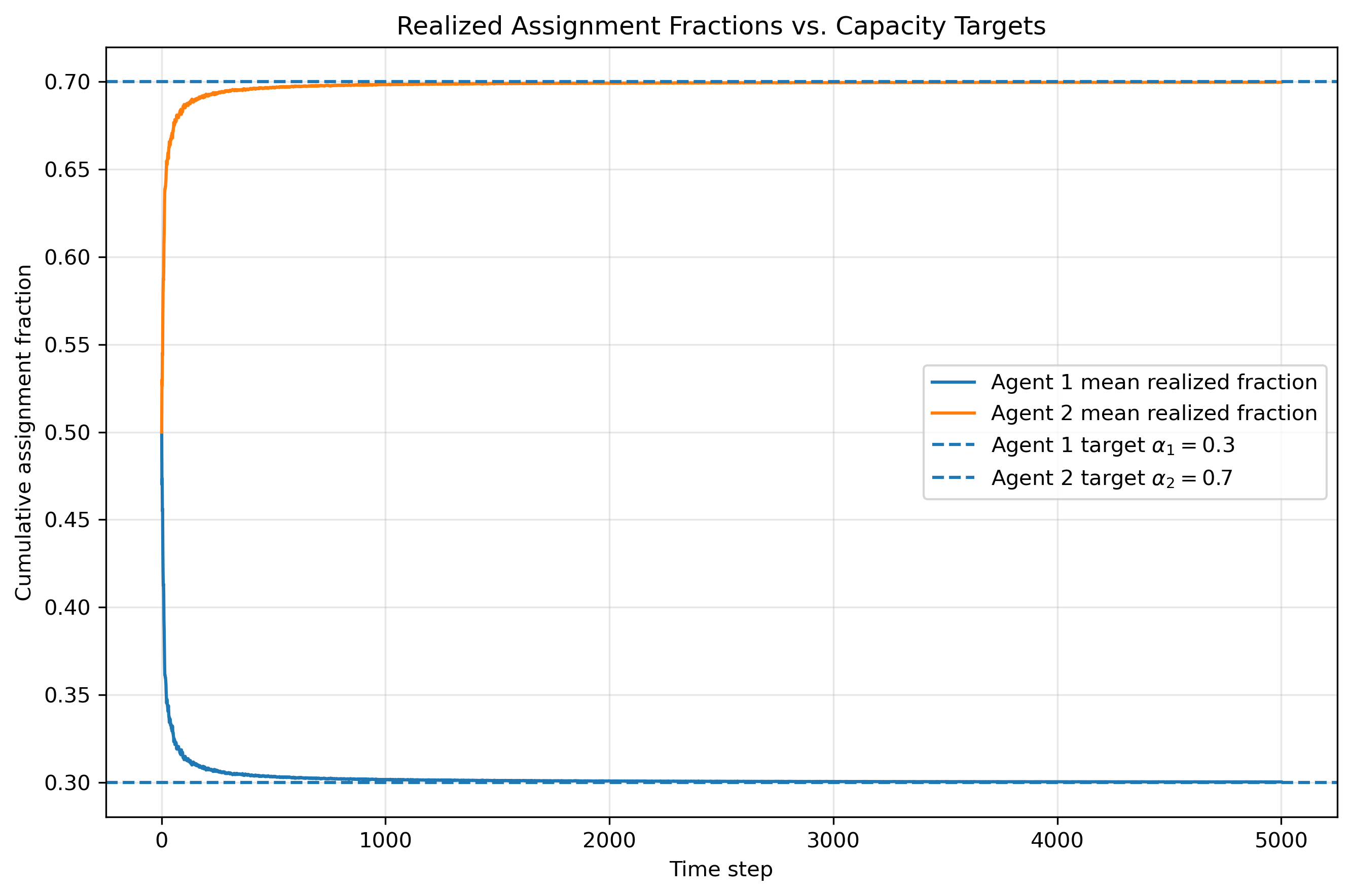}
        \caption{Realized assignment fractions versus capacity targets.}
        \label{fig:capacity_simulation_assignment}
    \end{subfigure}
    \hfill
    \begin{subfigure}[t]{0.49\textwidth}
        \centering
        \includegraphics[width=\linewidth]
        {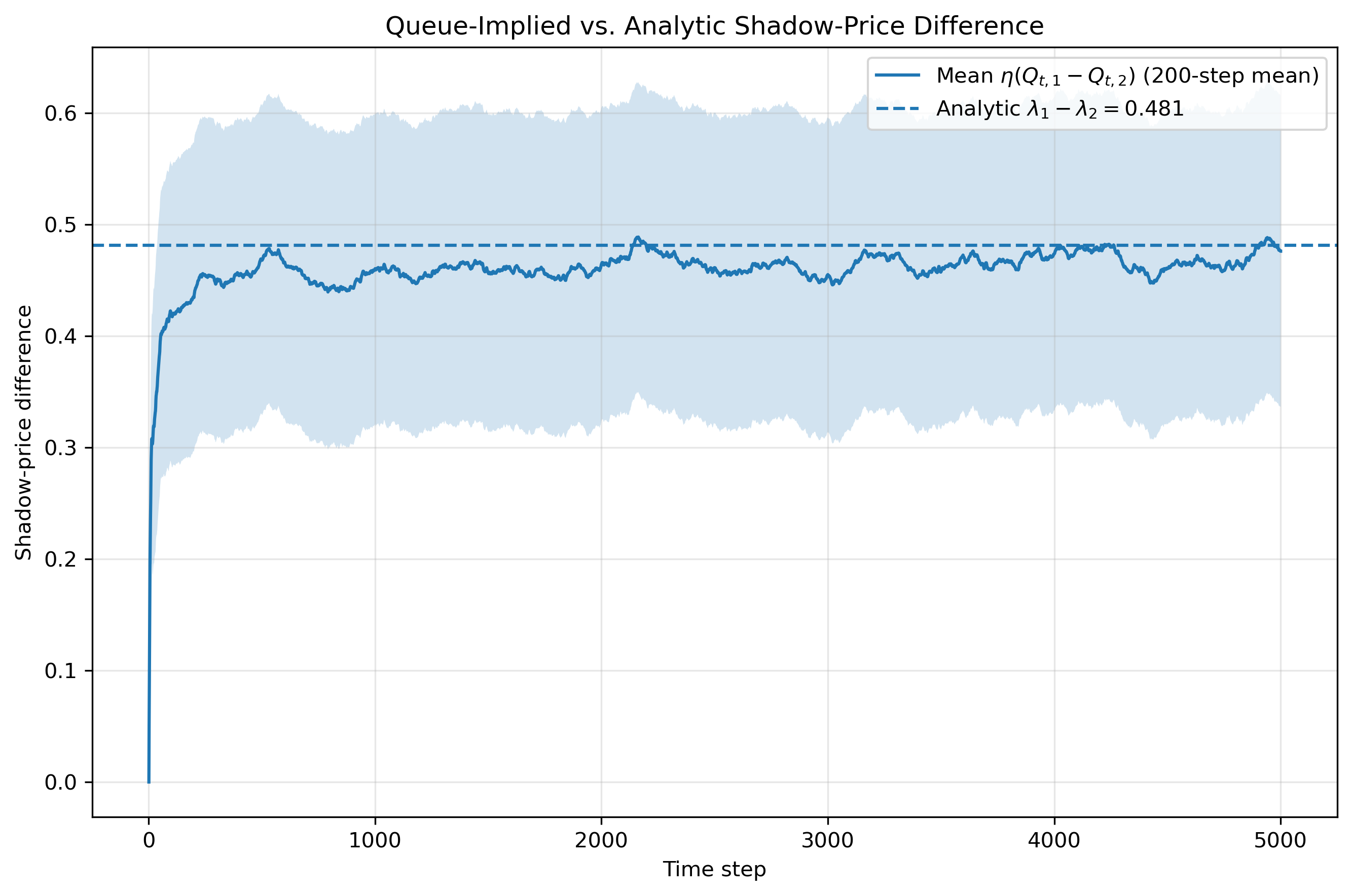}
        \caption{Queue-implied versus analytical shadow-price difference.}
        \label{fig:capacity_simulation_price}
    \end{subfigure}
    \caption{
    Synthetic validation of the queue-based capacity mechanism.
    Solid curves show means across 50 randomized task streams and shaded regions
    show 95\% confidence intervals. Dashed lines denote the analytical targets.
    The shadow-price trajectory in panel (b) uses a 200-step rolling mean.
    }
    \label{fig:capacity_simulation}
\end{figure}

\paragraph{Real-data validation.}
We further evaluate capacity enforcement on the Bank dataset using target
capacities $(\alpha_1,\alpha_2)=(0.2,0.8)$ and $\eta=0.5$ across 100 randomized
task streams. We define convergence as the first task at which the cumulative
realized assignment fraction is within $0.01$ of the target capacity and remains
within this tolerance for the subsequent 500 tasks.

As shown in Table~\ref{tab:bank_capacity}, all four policies satisfy this
criterion in all 100 runs, and none subsequently violates the $0.01$
tolerance. Median convergence occurs within the first 158 tasks for every
policy, or at most approximately $1.3\%$ of the task horizon. The maximum
final absolute deviation from the target capacity across all runs is only
$0.000129$.

We repeat this evaluation on all other datasets under the same target capacities and $\eta$.
Table~\ref{tab:capacity_all_datasets} reports, for all experiments, the worst case across the four contextual policies of the realized deviation from the target capacity, averaged over the final 5\% of each horizon. The
deviation is at most $0.0113$ across every dataset and policy, and below $0.001$ for all datasets except MMLU, whose much shorter horizon ($T=217$, versus $T\geq2{,}000$ for the others) leaves less room for the running assignment fraction to average out. Together with the synthetic results above, these results show that the virtual queue mechanism closely enforces the desired capacity targets across experiments.

\begin{table}[t]
    \centering
    \caption{
    Capacity enforcement on the Bank dataset across 100 randomized task streams.
    Convergence is defined as entering a $\pm0.01$ neighborhood of the target
    capacity and remaining within it for the subsequent 500 tasks.
    }
    \label{tab:bank_capacity}
    \small
    \begin{tabular}{lccc}
        \toprule
        Policy
        & Median conv.
        & Mean conv.
        & Mean conv. (\% $T$) \\

        \midrule
        Logistic TS     & 158 & 156 & 1.26\% \\
        Tree TS         & 153 & 156 & 1.26\%  \\
        Logistic Greedy & 132 & 128 & 1.04\%  \\
        Tree Greedy     &  96 & 103 & 0.84\%  \\
        \bottomrule
    \end{tabular}
\end{table}

\begin{table}[t]
    \centering
    \caption{
    Capacity enforcement across datasets, reporting the worst case over the four contextual policies (Logistic TS, Tree TS, Logistic Greedy, Tree Greedy). We report the mean and maximum, across 100 randomized task streams, of the absolute deviation between the realized assignment
    fraction and the target capacity $\alpha_1=0.2$, averaged over the final
    5\% of each dataset's horizon.
    }
    \label{tab:capacity_all_datasets}
    \small
    \begin{tabular}{lcc}
        \toprule
        Dataset & Mean final deviation & Max final deviation \\
        \midrule
        Bank        & 0.0001 & 0.0001 \\
        Cardio      & 0.0001 & 0.0001 \\
        Credit      & 0.0002 & 0.0002 \\
        Coupon      & 0.0004 & 0.0005 \\
        Camelyon17  & 0.0009 & 0.0011 \\
        ImageNet16H & 0.0006 & 0.0009 \\
        Civil & 0.0003 & 0.0004 \\
        MMLU        & 0.0072 & 0.0113 \\
        \bottomrule
    \end{tabular}
\end{table}

\subsection{Robustness to the Feature-Split Choice}
\label{app:feature_split_robustness}

The tabular experiments in Section~\ref{subsec:main_results} construct two agents by randomly partitioning each dataset's features: Agent 1 is a logistic regression model trained on a subset of up to 20\% of the features, and Agent 2 is an XGBoost classifier trained on the remainder. To verify that the
measured gains of contextual assignment over the non-contextual (random) baseline are not an artifact of this particular 20/80 split, we sweep the
fraction of features assigned to Agent 1 over
$\{0.1, 0.2, 0.5, 0.8, 0.9\}$ and, for each fraction, over 5 independent random seeds governing which specific features are assigned to which agent, yielding 25 distinct agent configurations on the \textbf{Bank} dataset. For
each configuration we evaluate all four contextual policies (Logistic TS, Tree TS, Logistic Greedy, Tree Greedy) at $\alpha_1=\alpha_2=0.5$ and
$\eta=0.5$, and report the gain over random assignment.

Table~\ref{tab:feature_split_sweep} reports the mean, minimum, and maximum gain over random assignment across 5 seeds and 4 policies for each split fraction (20 evaluations per row). The gain is positive in all 100 evaluations, confirming that the improvement is not specific to the 20/80 split used in Section~\ref{subsec:main_results}. Consistent with Appendix~\ref{app:two_agent_disagreement}, the gain correlates with the realized disagreement mass $p_1+p_2$ at $r=0.80$--$0.87$ across policies, suggesting that the split fraction affects the gain through the disagreement it induces.

\begin{table}[t]
    \centering
    \caption{
    Gain over the capacity-matched random baseline as a function of the
    feature-split fraction assigned to Agent~1, on the \textbf{Bank} dataset.
    Each row aggregates 5 feature-assignment seeds $\times$ 4 contextual
    policies (20 evaluations) at $\alpha_1=\alpha_2=0.5$, $\eta=0.5$.
    }
    \label{tab:feature_split_sweep}
    \small
    \begin{tabular}{lccc}
        \toprule
        Split ($\alpha_1$ features to Agent 1) & Mean gain & Min gain & Max gain \\
        \midrule
        0.1 & 0.0301 & 0.0045 & 0.0470 \\
        0.2 & 0.0429 & 0.0050 & 0.1072 \\
        0.5 & 0.0329 & 0.0051 & 0.0871 \\
        0.8 & 0.0096 & 0.0023 & 0.0172 \\
        0.9 & 0.0086 & 0.0027 & 0.0142 \\
        \bottomrule
    \end{tabular}
\end{table}


\newpage
\vskip 0.2in
\bibliography{reference}

\end{document}